\documentclass[11pt,reqno]{article}
\usepackage{amscd}
\usepackage{amsmath}
\usepackage{amssymb}
\usepackage{amsthm}
\usepackage{amsfonts}
\usepackage{enumitem}
\usepackage{pgf,tikz}
\usepackage{mathrsfs}
\usepackage{multicol}
\usepackage{etoolbox}
\usepackage{hyperref}
\usepackage{todonotes}
\usepackage{soul}
\usepackage[capitalise]{cleveref}
\patchcmd{\section}{\scshape}{\bfseries}{}{}
\newcommand{\Z}{\mathbb{Z}}
\newcommand{\N}{\mathbb{N}}

\newcommand{\F}{\mathbb{F}}

\newcommand{\bP}{\mathbb{P}}

\newcommand{\cA}{\mathcal{A}}

\newcommand{\cC}{\mathcal{C}}
\newcommand{\cF}{\mathcal{F}}
\newcommand{\cG}{\mathcal{G}}

\newcommand{\cL}{\mathcal{L}}
\newcommand{\cM}{\mathcal{M}}
\newcommand{\cS}{\mathcal{S}}

\newcommand{\ol}[1]{\overline{#1}}
\newcommand{\Fqn}{\F_{q^n}}
\newcommand{\Gkn}{\cG_q(k,n)}
\newcommand{\ds}{\textup{d}_{\rm{s}}}
\renewcommand{\d}{\textup{d}}
\newcommand{\Mod}[1]{\ \left(\mathrm{mod}\ #1\right)}

\newcommand{\sm}{\setminus}

\newtheorem{theorem}{Theorem}[section]
\newtheorem{proposition}[theorem]{Proposition}
\newtheorem{lemma}[theorem]{Lemma}

\newtheorem{corollary}[theorem]{Corollary}

\theoremstyle{definition}
\newtheorem{example}[theorem]{Example}
\newtheorem{definition}[theorem]{Definition}
\newtheorem{rem}[theorem]{Remark}

\theoremstyle{remark}

\numberwithin{equation}{section}
\newcounter{alp}
\newcounter{ara}
\newcounter{rom}

\newenvironment{alphalist}{\begin{list}{(\alph{alp})\hfill}{\usecounter{alp}
			\topsep0ex \labelwidth.6cm \leftmargin.6cm \labelsep0cm
			\rightmargin0cm \parsep0ex \itemsep0ex
			\partopsep0ex}}{\end{list}}
\newenvironment{arabiclist}{\begin{list}{(\arabic{ara})\hfill}{\usecounter{ara}
			\topsep0ex \labelwidth.7cm \leftmargin.7cm \labelsep0cm
			\rightmargin0cm \parsep0ex \itemsep0ex
			\partopsep1.6ex}}{\end{list}}
\DeclareMathOperator{\Orb}{Orb}
\DeclareMathOperator{\Stab}{Stab}
\DeclareMathOperator{\lcm}{lcm}

\reversemarginpar
\begin{document}
		
	\title{Distance Distributions of Cyclic Orbit Codes }
	\author{Heide Gluesing-Luerssen$^\ast$ and Hunter Lehmann\footnote{HGL was partially supported by the grant \#422479 from the Simons Foundation.
  HGL and HL are with the Department of Mathematics, University of Kentucky, Lexington KY 40506-0027, USA;
\{heide.gl, hunter.lehmann\}@uky.edu.}}

	\date\today
	\maketitle
	
	\begin{abstract}\label{sec:Abstract}
		The distance distribution of a code is the vector whose $i^\text{th}$ entry is the number of pairs of codewords with distance $i$.
		We investigate the structure of the distance distribution for cyclic orbit codes, which are subspace codes generated by the action of $\F_{q^n}^*$ on an $\F_q$-subspace $U$ of $\F_{q^n}$.
		We show that for optimal full-length orbit codes the distance distribution depends only on $q,\,n$, and the dimension of~$U$.
		For full-length orbit codes with lower minimum distance, we provide partial results towards a characterization of the distance distribution, especially in the case that any two codewords
		intersect in a space of dimension at most 2.
		Finally, we briefly address the distance distribution of a union of optimal full-length orbit codes.
	\end{abstract}
	
	%%%%%%%%%%%%%%%%%%%%%%%%%%%%%%%%%%%%%%%%%%%%%%%%%%%%%%
	\section{Introduction}\label{sec:Introduction}
	%%%%%%%%%%%%%%%%%%%%%%%%%%%%%%%%%%%%%%%%%%%%%%%%%%%%%%
	
	Following the seminal work of K\"{o}tter and Kschischang \cite{KoKsch08} in 2008, there have been a variety of lines of research on subspace codes and their
	applications to random network coding.
	Two major directions stand out: attempts to maximize the size of a subspace code given a fixed ambient space and minimum distance, and attempts to find algebraic constructions
    of subspace codes with best possible minimum distances, see
    \cite{HKK14,SiTr15,GLT16,HKKW16,HeKu17,CKMP19} for some of the more recent papers as well as the monograph~\cite{GPSV18} on network coding and subspace designs.
	One class of subspace codes that have attracted particular interest are cyclic orbit codes \cite{EtVa11,TMBR13,GLMT15,BEGR16,OtOz17,ChenLiu18,ZhTa19}
	due to their algebraic structure and efficient encoding/decoding algorithms.
	
	In this paper, we are interested in further classifying cyclic orbit codes of a fixed size and minimum distance using the finer invariant of the distance distribution.
	The latter encodes, for any possible subspace distance, the number of codeword pairs with that distance.
	It can thus detect subspace codes with the fewest number of codeword pairs attaining the minimum distance.
	Such codes may be regarded as superior to those with the same minimum distance but with more codeword pairs attaining that distance.
	
	A cyclic orbit code is a subspace code of the form $\Orb(U)=\{\alpha U\mid \alpha\in\F_{q^n}^*\}$, where $U$ is an $\F_q$-subspace of the field extension $\F_{q^n}$.
	In particular it is a constant-dimension code, that is, all subspaces in the code have the same dimension, namely $k:=\dim(U)$.
	It is well known that if $\Orb(U)$ has maximum possible distance, i.e.~$2k$, then~$k$ is a divisor of~$n$ and $U$ is a shift of the subfield $\F_{q^k}$.
	These codes are known as spread codes and their distance distribution is trivial because all subspaces intersect pairwise trivially.
	Their downside is their small size: they contain only $(q^n-1)/(q^k-1)$ codewords, which is the smallest size of any cyclic orbit code generated by a $k$-dimensional subspace.
	On the other hand, the largest size of such a code is $(q^n-1)/(q-1)$, and codes attaining this size will be called full-length orbit codes.
	Full-length orbit codes with distance $2k-2$, which is the best possible, will be called optimal full-length orbit codes.
	Hence optimal full-length orbit codes maximize the size of the code as well as the distance (as long as the latter is less than $2k$).
	
	Over the last few years, several different constructions of optimal full-length orbit codes have been found \cite{BEGR16,ChenLiu18,OtOz17}.
	In~2018, Roth, Raviv, and Tamo~\cite{RRT18} showed that all of these codes are generated by subspaces known as Sidon spaces.
	Our first major result,~\cref{theo:SidonCount}, shows that the distance distribution of optimal full-length orbit codes is fully determined by the parameters $q,\,n,$ and $k$,
	regardless of the choice of Sidon space.
	In deriving this result, another interesting parameter arises, namely the number, $f(U)$, of fractions inside the field $\F_{q^n}$ that can be obtained from elements of the given subspace~$U$
	(up to factors from~$\F_q$).
	For Sidon spaces this number is fully determined by $q,\,n,\,k$.
	Furthermore, we provide the minimum and maximum possible value of $f(U)$ over all $k$-dimensional subspaces and show that
	$f(U)$ is minimal iff $\Orb(U)$ is a spread code and maximal iff $\Orb(U)$ is an optimal full-length orbit code.
	
	In \cref{sec:GeneralOrbits}, we investigate the distance distribution of full-length orbit codes  with distance less than $2k-2$.
	In this case, the distance distribution is -- unsurprisingly -- not fully determined by $q,\,n,\,k$ and the distance.
	In \cref{thm:Dist2k4L2Structure} we describe the distance distribution as closely as possible for the case where the distance is $2k-4$.
	It involves, in addition to $q,\,n,$ and $k$, a further parameter~$r$, whose meaning will become clear in  \cref{sec:GeneralOrbits}.
	Various examples illustrate possible values of this parameter, but more work is needed to find its exact range or at least bounds.
	Alternatively, the distance distribution is fully determined by $q,\,n,\,k$ and the above mentioned parameter $f(U)$.
	However, we do not yet understand what values $f(U)$ may take or how to design subspaces with a particular value.
	
	Finally, in \cref{sec:UnionOfOrbits} we consider codes that are the union of optimal  full-length orbit codes.
	Constructions of such codes can be found in \cite{RRT18}.
	We show that \cref{theo:SidonCount}  generalizes straightforwardly to this scenario, that is, the distance distribution is fully determined by $q,\,n,\,k$, and the number of orbits in the union.
	
	Throughout the paper, we will in fact study the intersection distribution rather than the distance distribution.
	That is, we count the number of codeword pairs whose intersection attains a given dimension.
	Thanks to the definition of the subspace distance in~\eqref{eq:ds} this is clearly equivalent to studying the distance distribution.
	
	%%%%%%%%%%%%%%%%%%%%%%%%%%%%%%%%%%%%%%%%%%%%%%%%%%%%%%
	\section{Preliminaries}\label{sec:Preliminaries}
	%%%%%%%%%%%%%%%%%%%%%%%%%%%%%%%%%%%%%%%%%%%%%%%%%%%%%%
	
	We begin by recalling some basic facts about subspace codes and cyclic orbit codes.
	Throughout we fix a finite field~$\F_q$.
	A \emph{subspace code (of block length~$n$)} is simply a collection of subspaces in~$\F_q^n$ with at least two elements.
	The code is called a \emph{constant-dimension code} if all subspaces have the same dimension.
	The \emph{distance between two subspaces} $V,W\subseteq\F_q^n$ is defined as
	\begin{equation}\label{eq:ds}
	\d(V,W):=\dim V+\dim W-2\dim(V\cap W)
	\end{equation}
	and the \emph{subspace distance} of a subspace code~$\cC$ is
	$\ds(\cC):=\min\{\d(V,W)\mid V,\,W\in\cC,\,V\neq W\}$.
	
	Cyclic orbit codes are most conveniently defined in the field extension $\Fqn$, considered as an $n$-dimensional $\F_q$-vector space.
	Let $\Gkn$ be the Grassmannian of $k$-dimensional $\F_q$-subspaces of $\F_{q^n}$.
	Then $\F_{q^n}^*$ induces a group action on $\Gkn$ via $(\alpha,U)\longmapsto \alpha U$, where
	$\alpha U=\{\alpha u\mid u\in U\}$, which of course is a subspace in $\Gkn$.
   A constant-dimension codes in $\Gkn$ is called a \emph{cyclic subspace code} if it is invariant under this group action.
    Hence cyclic subspace codes are unions of orbits under this action.
    Throughout most of this paper we will study codes that form a single orbit and only in \cref{sec:UnionOfOrbits} turn to more general cyclic subspace codes.
	We fix the terminology of cyclic orbit codes and list some properties in the next definition.
	Further details can be found in~\cite{GLMT15}.
	
	%%%%%%%%%%%%%%%%%%%
	\begin{definition}\label{def:COC}
		Let $U\in\Gkn$.
		The \emph{cyclic orbit code generated by~$U$} is the subspace code $\Orb(U):=\{\alpha U \mid \alpha\in\Fqn^* \}$.
		It is a constant-dimension code of dimension~$k$.
		We define the \emph{stabilizer} of~$U$ in the obvious way as $\Stab(U):=\{\alpha\in\Fqn^*\mid \alpha U=U\}$.
		It is easy to see that $\Stab(U)=\F_{q^t}^*$ for some $t\in\N$ (which is a divisor of~$\gcd(k,n)$).
		In fact, the field~$\F_{q^t}$ is the largest subfield of $\Fqn$ over which~$U$ is a vector space, i.e.,~$U$ is closed under multiplication by scalars in~$\F_{q^t}$.
		The orbit-stabilizer theorem tells us that if $\Stab(U)=\F_{q^t}^*$ then $|\Orb(U)|=(q^n-1)/(q^t-1)$.
		If $t=1$, we call $\Orb(U)$ a \emph{full-length orbit code}.
	\end{definition}
	%%%%%%%%%%%%%%%%%%%%%%
	
	Taking the stabilizer into account, we can give a more concise description of the orbit.
	In order to do so, we make the following definition.
	It appeared first in \cite[Def.~4]{BEGR16}.
	
	%%%%%%%%%%%%%%%%%%
	\begin{definition}\label{def:SimFt}
		Let $t\in\N$ be a divisor of~$n$. On $\Fqn^*$ we define the equivalence relation
		\[
		\alpha\sim_t\beta \Longleftrightarrow \frac{\alpha}{\beta}\in \F_{q^t}.
		\]
		Note that the right hand side is equivalent to $ \alpha\F_{q^t}^*=\beta\F_{q^t}^*$.
		We set $\mbox{$\F_{q^n}^*/_{\sim_t}$} =\bP_t(\Fqn)$, the \emph{projective space} over the $\F_{q^t}$-vector space $\F_{q^n}$.
		Clearly $|\bP_t(\F_{q^n})|=\frac{q^n-1}{q^t-1}$.
		The equivalence class of~$\alpha\in\Fqn^*$ is denoted by $\ol{\alpha}^{(t)}$.
		For $t=1$ we omit the subscript/superscript; thus $\alpha\sim\beta \Longleftrightarrow \alpha\beta^{-1}\in \F_q$ and
		$\bP(\Fqn)=\bP_1(\Fqn)=\{\ol{\alpha}\mid \alpha\in\Fqn\}$.
	\end{definition}
	%%%%%%%%%%%%%%%
	
	Note that $\ol{\alpha}^{(t)}=\alpha\F_{q^t}^*$ and the projective space $\bP_t(\Fqn)$ is actually the cyclic orbit code generated by
	$\F_{q^t}$ if we add the zero vector to every equivalence class $\alpha\F_{q^t}^*$.

	%%%%%%%%%%%%%%
	\begin{rem}\label{rem:SizeOrb}
		Let $U\in\Gkn$ and $\Stab(U)=\F_{q^t}^*$. Then the map
		\[
		\bP_t(\F_{q^n})\longrightarrow\Orb(U),\ \ol{\alpha}^{(t)}\longmapsto \alpha U
		\]
		is a well-defined bijection.
	\end{rem}
	%%%%%%%%%%%%%%%%
	
	Let us now turn to the minimum distance and the distance distribution of a cyclic orbit code.
	Fix a subspace $U\in\Gkn$ and let $\Stab(U)=\F_{q^t}^*$.
	By the very definition of the subspace distance in~\eqref{eq:ds} we have
	$\d(\beta U,\alpha U)=\d(U,\alpha\beta^{-1}U)$ for all $\alpha,\beta\in\Fqn^*$.
	Furthermore, $\d(U,\alpha U)=2k-2\dim(U\cap\alpha U)$ for any $\alpha\in\Fqn^*$.
	This implies
	\begin{align*}
	\ds(\Orb(U))&=\min\{\d(U,\alpha U)\mid \alpha\in\Fqn^*,\,\alpha U\neq U\}\\
	     &=2k-2\max\{\dim(U\cap\alpha U)\mid \alpha\in\Fqn^*,\,\alpha U\neq U\}.
	\end{align*}
	
	In this paper we will study the distance distribution of cyclic orbit codes.
	Without loss of generality we may restrict ourselves to the case where $2k\leq n$.
    Indeed, because $\d(V^\perp, W^\perp)=\d(V,W)$, where $V^\perp$ denotes the orthogonal complement of $V$ with respect to the standard dot product, a subspace code
    and its dual have the same distance distribution.
	For the following definition recall that $\d(U,\alpha U)$ is always even and that
	$\d(U,\alpha U)=2i\Longleftrightarrow \dim(U\cap \alpha U)=k-i$.
	Furthermore, we projectify simply over the scalar field~$\F_q$ as this does not require knowledge of the stabilizer.
	Of course, $\Orb(U)=\{\alpha U\mid \ol{\alpha}\in\bP(\F_{q^n})\}$ and $\alpha U=\beta U$ iff $\ol{\alpha}^{(t)}=\ol{\beta}^{(t)}$.

	%%%%%%%%%%%%%%%%%%%
	\begin{definition}\label{def:DistDistr}
		Let $k\leq n/2$ and $U\in\Gkn$. Suppose $\ds(\Orb(U))=2d$.
        Set $\ell=k-d$, thus $\ds(\Orb(U))=2k-2\ell$ and $\ell=\max\{\dim(U\cap\alpha U)\mid a\in\F_{q^n}^*,\,U\neq\alpha U\}$.
        We call~$\ell$ the \emph{maximum intersection dimension of} $\Orb(U)$.
		For $i=d,\ldots,k$ define $\delta_{2i}=|\{\alpha U\mid \d(U,\alpha U)=2i\}|$.
		Then $(\delta_{2d},\ldots,\delta_{2k})$ is the \emph{distance distribution} of $\Orb(U)$.	
		For $i=0,\ldots,\ell$ we define $\lambda_i=|\cL_i|$, where
		\[
		\cL_i=\cL_i(U)=\{\ol{\alpha}\in\bP(\F_{q^n})\mid \dim(U\cap\alpha U)=i\},
		\]
		and call
		$(\lambda_0,\ldots,\lambda_\ell)$ the \emph{intersection distribution} of $\Orb(U)$.
		Suppose $\Stab(U)=\F_{q^t}^*$.
        Then $\lambda_i=(q^t-1)/(q-1)\delta_{2(k-i)}$ for $i=0,\ldots,\ell$ and
		$\sum_{i=d}^k\delta_{2i}=\sum_{i=0}^\ell(q-1)/(q^t-1)\lambda_i=|\Orb(U)\setminus\{U\}|=(q^n-1)/(q^t-1)-1$.
	\end{definition}
	%%%%%%%%%%%%%%%%%%
	
	A few comments are in order.
	First of all, in the distance distribution we only count the distances to the ``reference space''~$U$.
	This may be regarded as the analogue of the weight distribution of a linear block code where only the distances to the zero vector are counted as
	opposed to all pairwise distances.
	The complete number of pairs $(\beta U,\alpha U)$ such that $\d(\beta U,\alpha U)=2i$ is then $(q^n-1)(q^t-1)^{-1} \delta_{2i}$.
	Secondly, for the intersection distribution we count each single subspace~$\alpha U$ with a multiplicity $(q^t-1)(q-1)^{-1}=|\bP(\F_{q^t})|$,
    thus the factor relating~$\delta_{2(k-i)}$ and~$\lambda_i$.
    All of this shows that the intersection distribution, as defined above, fully determines the distance distribution.
    From now on we will study the intersection distribution.
    Accordingly, instead of the subspace distance $\ds(\Orb(U))=2d$ we will specify the parameter $\ell=k-d$, which we call the ``maximum intersection dimension''

    We collect a few properties and special cases in the following remark.

	%%%%%%%%%%%%%%
	\begin{rem}\label{rem:tmult}
		Consider the situation of \cref{def:DistDistr} where $\Stab(U)=\F_{q^t}^*$.	
       Then~$U$ and all its cyclic shifts are vector spaces over~$\F_{q^t}$, and hence $t\mid k$.
		As a consequence, $\d(U,\alpha U)=2k-2\dim(U\cap \alpha U)$ is a multiple of~$t$ for all $\alpha\in\Fqn$ and $\delta_j=0$ if $j\not\in r\Z$, where $r=\lcm(2,t)$.
		For the same reason $\cL_i=\emptyset$ if $t\nmid i$. This also implies that $t\mid\ell$, thus either $\ell\geq t$ or $\ell=0$.
		If $\ell=0$, then all subspaces of the orbit code intersect trivially, and thus their union consists of $(q^n-1)(q^t-1)^{-1}(q^k-1) +1$ elements.
		Since this number can be at most $q^n$, we conclude that $t=k$. Thus we have the following scenarios:
		\begin{alphalist}
			\item If $\ds(\Orb(U))=2k$, then $\Stab(U)=\F_{q^k}^*$ and thus $U=a\F_{q^k}$ for some $a\in\Fqn$.
            This is the best distance a cyclic orbit code can have, but comes at the cost of the length, which is just $(q^n-1)/(q^k-1)$, the
		shortest possible among all cyclic orbit codes of dimension~$k$.
		The code is a spread code, i.e., all subspaces intersect pairwise trivially and their union is the entire space.
            The intersection distribution is simply given by $\lambda_0=(q^n-q^k)/(q-1)$.
			\item If $\ds(\Orb(U))<2k$, then we even have the upper bound $\ds(\Orb(U))\leq 2(k-t)$.
			In particular, if $t=1$, the code is a full-length orbit, i.e., has maximal possible length $(q^n-1)/(q-1)$, and its distance is at most $2k-2$.
			Later in \cref{theo:SidonCount} we will see that all  full-length orbits with distance $2k-2$ have the same intersection distribution
			$(\lambda_0,\lambda_1)$.
		\end{alphalist}
	\end{rem}
	%%%%%%%%%%%%%%%%%%%

	Part~(a) tells us in particular that the intersection distribution of $\Orb(U)$ is fully determined for any $1$-dimensional subspace~$U$.
	Therefore, we may restrict ourselves to $k\geq 2$.
    The following class of orbit codes will be at the focus of the next section.

   %%%%%%%%%%%%%%%%%%%%%%%%%
   \begin{definition}\label{def:OptFL}
   A full-length orbit code with distance $2k-2$ is called an \emph{optimal full-length orbit}.
   \end{definition}
   %%%%%%%%%%%%%%%%%%%%%%%%%%%	

	%%%%%%%%%%%%%%%%%%%%%%%%%%%%%
	\section{\mbox{}\!\!\!\!The Intersection Distribution of Optimal Full-Length Orbit Codes}\label{sec:fractions}
	We fix $k,n\in\N$ such that $2\leq k\leq n/2$.
	We introduce some crucial parameters associated with a given subspace.
	They will be needed later to study the intersection distribution of cyclic orbit codes.

	%%%%%%%%%%%%%%%%%%%%
	\begin{definition}\label{def:UAssoc}
		Let $U\in\Gkn$ and $\ds(\Orb(U))=2k-2\ell$, thus~$\ell$ is the maximum intersection dimension of $\Orb(U)$.
		We define
		\begin{arabiclist}
			\item $\cL=\cL(U)=\bigcup_{i=1}^\ell\cL_i$, where $\cL_i$ is as in \cref{def:DistDistr}.
			\item $\cS=\cS(U)=\{\ol{\alpha}\in\bP(\Fqn)\mid \alpha\in\Stab(U)\}$. %; thus
			\item $\cM=\cM(U)=\{(\ol{u},\ol{v})\mid u,v\in U\setminus 0,\,\ol{u}\neq\ol{v}\}$.
			\item $\cF=\cF(U)=\{ \ol{uv^{-1}} \mid u,v \in U \sm0 \}$ and $f:=|\cF|$. We call~$\cF$ the \emph{set of fractions} of~$U$.
		\end{arabiclist}
       If $\Stab(U)=\F_{q^t}^*$, we have $s:=|\cS|=(q^t-1)/(q-1)$.
	\end{definition}
	%%%%%%%%%%%%%
	
	Note that $\cL=\{\ol{\alpha}\mid 0\neq U\cap\alpha U\neq U\}$.
	In particular $\cL\cap\cS=\emptyset$.
	Recall the intersection distribution $(\lambda_0,\ldots,\lambda_\ell)$, where $\lambda_i=|\cL_i|$.
	Since $\lambda_0$ is fully determined by $(\lambda_1,\ldots,\lambda_\ell)$, it suffices to study the sets $\cL_1,\ldots,\cL_\ell$.
	Hence we omit $\cL_0$ in the union $\cL$.
	As for part (3) above, note that for nonzero vectors $u,v\in U$ the property $\ol{u}\neq\ol{v}$ is equivalent to the linear independence of $u,v$ in the $\F_q$-vector space~$U$.
	Therefore,
	\begin{equation}\label{eq:Q}
	Q:=|\cM|=\frac{q^k-1}{q-1}\,\frac{q^k-q}{q-1}.
	\end{equation}
	Finally, the set~$\cF$ consists of all equivalence classes of fractions (within the field $\Fqn$) of nonzero elements in~$U$.
	Its size~$f$ will play a crucial role later on in the study of the intersection distribution of $\Orb(U)$.
	The next result shows the relation of~$\cF$ to $\Orb(U)$: the elements in the equivalence classes of~$\cF$ correspond to the shifts
    $\alpha U$ such that $\alpha U \cap U \neq 0$.
	In particular, for determining the intersection distribution we only need to consider shifts
	$\alpha U$ where $\ol{\alpha}\in\cF$, which reduces considerably the computational effort.

	%%%%%%%%%%%%%%%%%
	\begin{proposition}\label{prop:psiPreimageSize}
		Let $U\in\Gkn$ such that $\Stab(U)=\F_{q^t}^*$. Let $\ds(\Orb(U))=2k-2\ell$.
		Then the map
		$\psi: \cM \longrightarrow\cF,\quad (\ol{u},\ol{v})\longmapsto \ol{uv^{-1}}$ is well-defined and
		satisfies $\cF=\psi(\cM)\cup\{\ol{1}\}=\cL\cup\cS$.
		Furthermore,  for any $\ol{\alpha}\in\cF$ we have
		\begin{alphalist}
			\item $\ol{\alpha}\in\cS\Longleftrightarrow |\psi^{-1}(\ol{\alpha})|=(q^k-1)/(q-1)$,
			\item $\ol{\alpha}\in\cL_i\Longleftrightarrow |\psi^{-1}(\ol{\alpha})|=(q^i-1)/(q-1)$.
		\end{alphalist}
		Since $\cL_i=\emptyset$ if $t\nmid i$, this implies that the pre-images never have size $(q^i-1)/(q-1)$, where $t\nmid i$.
	\end{proposition}
	%%%%%%%%%%%%%%%%%%%%%%
	
	\begin{proof}
		The well-definedness of~$\psi$ is clear and so is $\psi(\cM)\cup\{\ol{1}\}=\cF$.
		To show that $\psi(\cM)\cup\{\ol{1}\}=\cL\cup\cS$,
		let $\ol{\alpha}=\ol{uv^{-1}}\in \psi(\cM)$ for some $u,v\in U\setminus 0$.
		Then there exists $\lambda\in\F_q^*$ such that $u=\lambda\alpha v$.
		Since $\lambda U=U$ this implies $U\cap \alpha U\neq 0$.
		Hence either $U=\alpha U$ or  $\dim(U\cap\alpha U)\in\{1,\ldots,\ell\}$.
		In the first case  $\alpha\in\F_{q^t}$, thus $\ol{\alpha}\in\cS$ and in the second case
		$\ol{\alpha}\in \cL$.
		Since obviously $\ol{1}\in\cS$, this shows $\psi(\cM)\cup\{\ol{1}\} \subseteq \cL\cup\cS$.
		The proof of the reverse inclusion proceeds similarly.
		If $\ol{\alpha}\in \cL\cup\cS$, then $1\leq\dim(U\cap\alpha U)\leq k$.
		Hence there exist $u,v \in U$ with $u=\alpha v$.
		So $\ol{\alpha}=\ol{uv^{-1}}$ is either~$\ol{1}$ or in $\psi(\cM)$.
		
		It remains to show~(a) and~(b).
		Note first that for $(\ol{v},\ol{w})\in\psi^{-1}(\ol{\alpha})$, the second component $\ol{w}$ is uniquely determined by the first one.
		Thus it suffices to count the number of possible first components.

		Let $\ol{\alpha}\in\cS$. Then $U=\alpha U$.
		As a consequence, for every $v\in U$ there exists a unique $w\in U$ such that $v=\alpha w$.
		Hence $\ol{\alpha}=\ol{vw^{-1}}=\psi(\ol{v},\ol{w})$.
		Because there exist $(q^k-1)/(q-1)$ elements $\ol{v}$ such that $v\in U$, the result follows.
		
		Let $\ol{\alpha}\in\cL_i$. Hence $\dim(U\cap \alpha U)=i$.
		Then for every $v\in U\cap\alpha U$ there exists $w\in U$ such that $v=\alpha w$.
		Using that there exist $(q^i-1)/(q-1)$ elements $\ol{v}$ such that $v\in U\cap\alpha U$, we may argue as above to conclude that
		$|\psi^{-1}(\ol{\alpha})|\geq(q^i-1)/(q-1)$.
		Conversely, suppose that $(\ol{x},\ol{y}) \in \psi^{-1}(\ol{\alpha})$.
		Then $\ol{x}=\ol{\alpha y}$ and $x=\lambda\alpha y$ for some $\lambda\in\F_q$.
		Thus $x\in U\cap\alpha U$.
		This leaves $(q^i-1)/(q-1)$ choices for $\ol{x}$, and thus $|\psi^{-1}(\ol{\alpha})|\leq(q^i-1)/(q-1)$.
		Hence $|\psi^{-1}(\ol{\alpha})|=(q^i-1)/(q-1)$.
	\end{proof}

	%%%%%%%%%%%%%%%%%
	\begin{corollary}\label{cor:Count}
		Let $U\in\Gkn$ such that $\Stab(U)=\F_{q^t}^*$. Let $\ds(\Orb(U))=2k-2\ell$.
		Recall the cardinalities $f=|\cF|,\,s=|\cS|,\,Q=|\cM|$, and $\lambda_i=|\cL_i|$ for $i=1,\ldots,\ell$.
		Then
		\begin{equation}\label{eq:fsum}
		f= s+\sum_{i=1}^\ell \lambda_i.
		\end{equation}
		and
		\begin{equation}\label{eq:Qsum}
		Q=\sum_{i=1}^{\ell}\frac{q^i-1}{q-1}\lambda_{i}+\frac{q^k-1}{q-1}(s-1).
		\end{equation}
		In the special case where $t=1$, i.e.~$\Orb(U)$ is a full-length orbit, we have $s=1$, and thus
		\begin{equation}\label{eq:t1sums}
		f= 1+\sum_{i=1}^\ell \lambda_i \ \text{ and }\
		Q=\sum_{i=1}^{\ell}\frac{q^i-1}{q-1}\lambda_{i}.
		\end{equation}
	\end{corollary}
	%%%%%%%%%%%%%%
	
	\begin{proof}
		The identity in~\eqref{eq:fsum} is a consequence of $\cL\cup\cS=\cF$  from \cref{prop:psiPreimageSize} along with $\cL\cap\cS=\emptyset$.
		From the same proposition we have $\psi(\cM)=\cL\cup\cS\setminus\{\ol{1}\}$, thus
		$\cM=\bigcup_{i=1}^\ell\psi^{-1}(\cL_i)\cup\psi^{-1}(\cS\setminus\{\ol{1}\})$.
		Now~\eqref{eq:Qsum} follows from \cref{prop:psiPreimageSize}(a) and~(b) and the cardinality of~$\cM$ in~\eqref{eq:Q}.
		The rest follows from $s=(q^t-1)/(q-1)$.
	\end{proof}
	
	Recalling~$Q$ from~\eqref{eq:Q}, the above identities~\eqref{eq:fsum} and~\eqref{eq:Qsum} allow us
	to give a lower and upper bound on the number of fractions of~$U$ in terms of $q$ and~$k$.
	
	%%%%%%%%%%%%%%%%%
	\begin{proposition}\label{prop:fbounds}
		With the data as in \cref{cor:Count} we have
		\[
		\frac{q^k-1}{q-1}\leq f\leq  Q+1.
		\]
		Furthermore,
		\begin{alphalist}
			\item $f=(q^k-1)/(q-1)\Longleftrightarrow\ell=0\Longleftrightarrow t=k$. This is the spread-code case, thus $U=a\F_{q^k}$ for some $a\in\Fqn$.
			\item $f=Q+1\Longleftrightarrow\ell=1\Longleftrightarrow\ds(\Orb(U))=2k-2$.
                    This is the case of optimal full-length orbits (see \cref{def:OptFL}).
		\end{alphalist}
	\end{proposition}
	%%%%%%%%%%%%%%%%%%%%%%%
	
	\begin{proof}
		The lower bound for~$f$ is obvious from $\dim(U)=k$.
		As for the upper bound, note that $\frac{q^i-1}{q-1} \geq 1$ for $i=1,\ldots,\ell$ and $i=k$.
		So by \eqref{eq:fsum} and \eqref{eq:Qsum}
		\begin{equation}\label{eq:fQineq}
		f-1 = \sum_{i=1}^\ell \lambda_i + (s-1)
		\leq \sum_{i=1}^\ell \left( \frac{q^i-1}{q-1} \right) \lambda_i + \frac{q^k-1}{q-1} (s-1) = Q.
		\end{equation}
		It remains to prove~(a) and~(b).
		(a) If $\ell=0$, then $t=k$ by \cref{rem:tmult} and $f=s=(q^k-1)/(q-1)$ thanks to~\eqref{eq:fsum} and \cref{def:UAssoc}.
		Conversely, let $f=(q^k-1)/(q-1)$.
		Then $f=|\{\ol{u}\mid u\in U\setminus 0\}|$.
		Replacing~$U$ by a suitable shift we may assume without loss of generality that $1\in U$.
		Then the above along with the definition $f=|\cF|=|\{\ol{uv^{-1}}\mid u,v\in U\setminus 0\}|$
        tells us that  for every $u,v\in U\setminus 0$ there exists $w\in U$ and $\lambda\in\F_q$ such that
		$uv^{-1}=\lambda w$.
		Hence $U$ is closed under division and inverses and is therefore the subfield $\F_{q^k}$.
		The rest follows again from \cref{rem:tmult}.
		
		(b) Suppose $\ell=1$, i.e., $\ds(\Orb(U))=2k-2$. Then \cref{rem:tmult}(b) implies $t=1$.
		Thus $s=1$ and subsequently $f=Q+1$ by~\eqref{eq:t1sums}.
		On the other hand, if $f=Q+1$ then we have equality in~\eqref{eq:fQineq}, which in turn implies
		$\ell=s=1$ since $\frac{q^i-1}{q-1} >1$ for $i>1$.
	\end{proof}
	
	Full-length orbits with maximum possible distance $2k-2$ do indeed exist.
This has been studied in detail in~\cite{RRT18}.
	
	%%%%%%%%%%%%%%%%%
	\begin{definition}[{\cite[Def.~1]{RRT18}}]\label{def:Sidon}
		A subspace $U\in\Gkn$ is called a \emph{Sidon space} if it has the property that whenever
		$a,b,c,d \in U\sm0$ are  such that $ab=cd$, then $\{\ol{a},\ol{b} \}=\{ \ol{c},\ol{d} \}$.
	\end{definition}
	%%%%%%%%%%%%%%%%%%

	%%%%%%%%%%%%%%%%%
	\begin{theorem}[{\cite[Lemma 34]{RRT18}}]\label{thm:Sidon}
		Let $U \in\Gkn$. Then  $\Orb(U)$ is a full-length orbit with minimum distance $2k-2$ if and only if~$U$ is a Sidon space.
	\end{theorem}
	%%%%%%%%%%%%%%%%%%%

	In \cite[Thm.~12 and Thm.~16]{RRT18} Sidon spaces in $\Gkn$ are constructed for the case where~$k<n/2$ is a divisor of~$n$
	or $k=n/2$ and $q\geq 3$, and thus the existence of full-length orbits with maximum possible distance is guaranteed for these cases.
	
	It follows now easily from the above results that all these orbit codes have the same intersection distribution.
	
	%%%%%%%%%%%%%%%%%%
	\begin{theorem}\label{theo:SidonCount}
		Let $U\in\Gkn$ be a Sidon space. Then the full-length orbit $\Orb(U)$ has intersection distribution $(\lambda_0,\lambda_1)$, where
		\[
		\lambda_1=Q=\frac{q^k-1}{q-1}\,\frac{q^k-q}{q-1}=f-1,\quad
		\lambda_0=\frac{q^n-1}{q-1}-\lambda_1-1.
		\]
		In particular, the intersection distribution of a full-length orbit code with maximum distance depends only on the parameters $q,n$, and $k$.
		Furthermore, all $k$-dimensional Sidon spaces in $\F_{q^n}$ have the same number of fractions.
	\end{theorem}
	%%%%%%%%%%%%%%%%%%%
	
	\begin{proof}
	Under the given assumptions we have $t=1$, and thus $s=1$, as well as $\ell=1$. Now the result for $\lambda_1$ follows from \eqref{eq:t1sums},
		while~$\lambda_0$ can be computed using \cref{def:DistDistr}.
	\end{proof}
	
    Since for $k=2$ every $U\in\Gkn$ leads to an orbit code with distance $\ds=2k$ or $\ds=2(k-1)$, we have fully described the intersection distribution of
    all such orbit codes.
    Hence from now on we may assume $k\geq 3$.

	Examples show that for full-length orbit codes with minimum distance at most $2(k-2)$ the intersection distribution does not only depend on $q,n,k$ and the minimum distance.
	We will study that case in further detail in the next section.

    We close this section with a generalization of the previous results by taking the stabilizer into account.
    This improves on the upper bound for $f$ given in \cref{prop:fbounds}.

    %%%%%%%%%%%%%%%%%%%%%%%
    \begin{proposition}\label{prop:GenSidon}
     Let $U\in\Gkn$ and $\Stab(U)=\F_{q^t}^*$. Then
     \[
       \frac{q^k-1}{q-1}\leq f\leq\frac{q^k-1}{q^t-1}\,\frac{q^k-q^t}{q-1}+\frac{q^t-1}{q-1}.
     \]
     \begin{alphalist}
     \item $f=\frac{q^k-1}{q^t-1}\,\frac{q^k-q^t}{q-1}+\frac{q^t-1}{q-1}\Longleftrightarrow\ds(\Orb(U))=2(k-t)$ (which is the maximum possible distance for
              a cyclic orbit code with stabilizer~$\F_{q^t}^*$).
              This is the case if and only if $U$ is a Sidon space over~$\F_{q^t}$, i.e., if $a,b,c,d\in U\setminus 0$ and $ab=cd$, then
             $\{\ol{a}^{(t)},\ol{b}^{(t)}\}=\{ \ol{c}^{(t)},\ol{d}^{(t)}\}$.
     \item If $\ds(\Orb(U))=2(k-t)$, then the intersection distribution is given by $(\lambda_0,\lambda_t)$, where
     		\[
     		\lambda_t=\frac{q^k-1}{q^t-1}\,\frac{q^k-q^t}{q-1}=f-\frac{q^t-1}{q-1},\quad
     		\lambda_0=\frac{q^n-q^t}{q-1}-\lambda_t.
     		\]
     \end{alphalist}
     In particular, if $\Stab(U)=\F_{q^{k/2}}^*$ then $f=(q^{3k/2}-1)/(q-1)$.
     \end{proposition}
    %%%%%%%%%%%%%%%%%%%%%%%%%

    \begin{proof}
    By assumption $t\mid \gcd(n,k)$. Set $\hat{q}=q^t,\, \hat{k}=k/t,$ and $\hat{n}=n/t$.
    Then $|\Orb(U)|=(\hat{q}^{\hat{n}}-1)/(\hat{q}-1)$, and thus it is a full-length orbit if considered as a collection of $\F_{\hat{q}}$-subspaces in the ambient space
    $\F_{\hat{q}^{\hat{n}}}$.
    Hence we may apply~\eqref{eq:t1sums} if we replace $\F_q$ by~$\F_{\hat{q}}$.
    In order to do so, we need to generalize \cref{def:UAssoc} and the sets~$\cL_i$ by projectifying with respect to the scalar field~$\F_{\hat{q}}$.
    We denote the resulting sets and cardinalities with a superscript $(t)$, thus $\cL_i^{(t)}=\{\ol{\alpha}^{(t)}\mid \dim_{\F_{\hat{q}}}(U\cap \alpha U)=i\}$ etc.
    Then $\cS^{(t)}=\{\ol{1}^{(t)}\}$ and
    \begin{equation}\label{eq:rel}
      \lambda_{it}=\frac{q^t-1}{q-1}\lambda_i^{(t)},\quad f=\frac{q^t-1}{q-1}f^{(t)},\quad s=\frac{q^t-1}{q-1}s^{(t)},\quad
      Q^{(t)}=|\cM^{(t)}|=\frac{\hat{q}^{\hat{k}}-1}{\hat{q}-1}\,\frac{\hat{q}^{\hat{k}}-\hat{q}}{\hat{q}-1}.
    \end{equation}
    Now we can prove the above statement.
     From \cref{prop:fbounds} we have $(\hat{q}^{\hat{k}}-1)/(\hat{q}-1)\leq f^{(t)}\leq Q^{(t)}+1$, and~\eqref{eq:rel} leads to the stated inequalities.
    \\
    (a) \cref{prop:fbounds}(b) tells us that $f^{(t)}= Q^{(t)}+1$ iff the subspace distance is $2(\hat{k}-1)$, and where the distance is computed via
    dimensions over~$\F_{q^t}$.
    Hence the latter becomes $\ds(\Orb(U))=2t(\hat{k}-1)=2(k-t)$ as dimensions over~$\F_q$.
   \\
   (b) From \cref{theo:SidonCount} we have $\lambda_1^{(t)}=Q^{(t)}=f^{(t)}-1$ and $\lambda_0^{(t)}=(\hat{q}^{\hat{n}}-1)/(\hat{q}-1)-\lambda_1^{(t)}-1$.
   Using~\eqref{eq:rel}, we obtain the stated expression for~$\lambda_t$ and~$\lambda_0$.

   Finally, if $t=k/2$, then $U$ is not a cyclic shift of a field and thus must be Sidon over $\F_{q^t}$.
   Hence we may apply~(a) and simplify.
     \end{proof}
	
	%%%%%%%%%%%%%%%%%%%%%%%%%%%%%%%%%%%%%%%%%%%%%%%%%%%%%%
	\section{Intersection Distribution of General Full-Length Orbit Codes}\label{sec:GeneralOrbits}
	%%%%%%%%%%%%%%%%%%%%%%%%%%%%%%%%%%%%%%%%%%%%%%%%%%%%%%
	
	In this section we will generalize some of the results of the previous section to the intersection distribution of a cyclic orbit code with smaller minimum distance.
	After the spread codes and optimal full-length orbits, the cyclic orbit codes with the best combination of orbit size and minimum distance
    are full-length orbit codes with minimum distance $2k-4$.
	Our goal in this section is to describe the intersection distribution of such codes in terms of the parameters $q,n,k$ and a new parameter~$r$.
	This new parameter counts the number of cyclic orbits generated by the $2$-dimensional intersections $U\cap \alpha U$.
	These parameters together with \eqref{eq:Qsum} are enough to completely determine the intersection distribution, which we give in the following theorem.
    In this section we may and will assume $3\leq k\leq n/2$.
	Recall $Q=|\cM|$ from~\eqref{eq:Q}.
	
	%%%%%%%%%%%%%%%%%%%%%%%%%%%%%%%%%%
	\begin{theorem}\label{thm:Dist2k4L2Structure}
		Let $U \in \cG_q(k,n)$ generate a full-length orbit with $\ds(\Orb(U))=2k-4$.
		Then one of the following cases occurs.
		\begin{alphalist}
		\item $U$ contains a cyclic shift of $\F_{q^2}$ (hence~$n$ is even). In this case $\Orb(U)$ has intersection distribution $(\lambda_0,\lambda_1,\lambda_2)$, where
			\begin{align*}
				\lambda_2 &= q+rq(q+1), \\
				 \lambda_1 &= Q-(q+1)\lambda_2=\frac{q^k-1}{q-1}\frac{q^k-q}{q-1} - (q+1)(q+rq(q+1)),\\
				 \lambda_0 &= |\bP(\F_{q^n})|-\lambda_1-\lambda_2-1= \frac{q^n-1}{q-1} +q^2(1+r(q+1))-Q-1
			\end{align*}
			for some $r\geq 0$.
		
		\item $U$ does not contain a cyclic shift of $\F_{q^2}$. In this case, $\Orb(U)$ has intersection distribution $(\lambda_0,\lambda_1,\lambda_2)$, where
			\begin{align*}
				\lambda_2 &= rq(q+1), \\
				\lambda_1 &= Q-(q+1)\lambda_2=\frac{q^k-1}{q-1}\frac{q^k-q}{q-1} - rq(q+1)^2,\\
				\lambda_0 &= |\bP(\F_{q^n})|-\lambda_1-\lambda_2-1= \frac{q^n-1}{q-1} +rq^2(q+1)-Q-1
			\end{align*}
			for some $r \geq 1$.
		\end{alphalist}
	\end{theorem}
	%%%%%%%%%%%%%%%%%%%%%%%%%

	The proof is deferred to the end of this section.
	Before setting up the necessary preparation, we draw some further conclusions about the possible values of $\lambda_1$ and $\lambda_2$.
	
	%%%%%%%%%%%%%%%%%%%%%%%%%%
	\begin{corollary}\label{cor:2k-4DistIneq}
		Let $U\in \cG_q(k,n)$ generate a full-length orbit with $\ds(\Orb(U))=2k-4$.
		Then the intersection distribution $(\lambda_0,\lambda_1,\lambda_2)$ of $\Orb(U)$ depends only on $q,n,k$, and $f$.
		Further, the following inequalities hold.
		 \begin{equation}\label{eq:Ineq}
		    q \leq \lambda_2 \leq \frac{Q}{q+1}, \qquad 0\leq \lambda_1 \leq Q-q(q+1), \qquad \frac{Q}{q+1} \leq f-1 \leq Q -q^2.
		 \end{equation}
	\end{corollary}
	%%%%%%%%%%%%%%%%%%%%%%%%

	\begin{proof}
		By assumption $\Stab(U)=\F_q^*$, and thus $t=s=1$ in Corollary~\ref{cor:Count}. Hence \cref{eq:t1sums} reduces to
		\begin{equation}\label{eq:fQsumsDist2k4}
		f-1=\lambda_1+\lambda_2 \qquad Q=\lambda_1+ (q+1)\lambda_2.
		\end{equation}
		Because either $\frac{q^k-1}{q-1}$ or $\frac{q^{k-1}-1}{q-1}$ is divisible by $q+1$, we have $Q\in q(q+1)\Z$; see~\eqref{eq:Q}.
		Since \cref{thm:Dist2k4L2Structure} says that $\lambda_2 \in q\Z$, we also have $\lambda_1 \in q(q+1)\Z$ and $(f-1)\in q\Z$.
		In fact \cref{thm:Dist2k4L2Structure} implies the stronger statement that $f-1\in q(q+1)\Z$ if and only if $U$ does not contain any cyclic shift of $\F_{q^2}$.
		Now we can solve this system of equations for $\lambda_1$ and $\lambda_2$ in terms of $q,Q$, and $f$ to get
		\begin{equation}\label{eq:flambda}
			\lambda_1=\frac{1}{q}((q+1)(f-1)-Q) \qquad \lambda_2=\frac{1}{q}(Q-(f-1)).
		\end{equation}
		Both of these values are guaranteed to be in $\Z$ by the above discussion.
		Therefore the intersection distribution of $\Orb(U)$ is completely determined by $q,n,k$ and the value $f=|\cF(U)|$.
        The inequalities of \eqref{eq:Ineq} follow from $\lambda_1\geq 0$ and $\lambda_2 \geq q$ together with~\eqref{eq:fQsumsDist2k4}.
	\end{proof}
	
	The next example shows that equality can be achieved on both sides of~\eqref{eq:Ineq}, with the maximum of~$\lambda_2$ corresponding to the minimum of~$\lambda_1$  and vice versa.
	In other words, there exist subspaces~$U$ where $\ell=2$, $t=1$, and $\dim(U\cap \alpha U) \in \{0,2,k\}$ for all $\alpha \in \F_{q^n}^*$, and hence $\lambda_1=0$.
	Similarly, there exist subspaces~$U$ with $\ell=2, t=1$, and $\lambda_2=q$.
	However, in general the restriction that $\lambda_2 \Mod{q(q+1)} \in \{0,q\}$ means that the upper bound of $Q/(q+1)$ may not be attainable.
	
	%%%%%%%%%%%%%%%%%%%%%%%%%
	\begin{example}\label{ex:NonTrivStab2k4}
		Let $q=3, k=3, n=8$ and let $\gamma$ be primitive in $\F_{3^8}$.
        Then  $\alpha=\gamma^{\frac{3^8-1}{3^2-1}}$ is a  primitive element of $\F_{9} \subseteq \F_{3^8}$ and $\beta=\gamma^{\frac{3^8-1}{3^4-1}}$ is a primitive element of
        $\F_{81} \subseteq \F_{3^8}$.
		Define $U=\langle 1,\alpha, \rho \rangle$ for some $\rho \in \F_{3^8}^*\sm\langle 1,\alpha \rangle$.
		Hence $\F_9\subseteq U$.
		 Since $\gcd(k,n)=1$, the subspace~$U$ generates a full-length orbit for any linearly independent choice of $\rho$.
		There are two possibilities:
		\begin{alphalist}
			\item $\rho \in \F_{81}\sm\F_{9}$
			
			\item $\rho \in \F_{3^8}\sm\F_{81}$.
		\end{alphalist}
		In (a), we have $U\subseteq \F_{81}$, and thus $\frac{v}{w} \in \F_{81}$ for any $v,w \in U\sm0$.
       As a consequence, the only nonzero intersections are of the form $U\cap \beta^s U$.
		Since~$U$ and $\beta^s U$ are both 3-dimensional $\F_3$-subspaces of the 4-dimensional $\F_3$-subspace $\F_{81}$, we conclude $\dim(U\cap\beta^s U)\in \{2,3\}$.
		This leads to
		\[
               \lambda_1=0, \qquad \lambda_2=\frac{Q}{q+1}=39
        \]
		for any such choice of $\rho$.
		The~$39$ elements resulting in a $2$-dimensional intersection are exactly the elements of $\bP(\F_{81})\sm\{\ol{1}\}$.
		\\
		In (b), a computation using SageMath shows that
		\[
                      \lambda_1=Q-q(q+1) =144, \qquad \lambda_2=q=3
        \]
		for any such choice of $\rho$.
		Noting that $\alpha^s \F_9 = \F_9$ for any~$s$ and $\Stab(U)=\F_3$, we conclude that $\F_9 \subseteq U \cap \alpha^s U\subsetneq U$ for any~$s$ such that
       $\alpha^s\not\in\F_3$.
        It follows that $U \cap \alpha^s U=\F_9$ and therefore the three elements that result in a 2-dimensional intersection are exactly the elements of $\bP(\F_{9})\sm\{\ol{1}\}$.
		We will see later in \cref{prop:2k4IntReprNonTrivStab} that this can be generalized.
	\end{example}

	The construction in part (a) of this example generalizes to provide examples of full-length orbit codes where the subspace distance is $2k-2\ell$ for arbitrarily large $\ell$ and $\lambda_i=0$ for small $i$.
	Since computation with SageMath shows that many, if not most, subspaces have $\lambda_1 > \lambda_i$ for $i>1$, these are unusual subspaces.
	The following example generalizes \cref{ex:NonTrivStab2k4}(a).
	
	\begin{example}
		Let $q$ be a prime power and consider a  tower of fields $\F_q \subset \F_{q^s} \subset \F_{q^m} \subset \F_{q^n}$.
		Our goal is to find $U$ so that $U$ has full length orbit but $\lambda_i=0$ for small values of $i$, generalizing \cref{ex:NonTrivStab2k4} (a).
		To this end, let $U$ be a $k$-dimensional $\F_q$-subspace of $\F_{q^m}$ containing $\F_{q^s}$ and that is not an $\F_{q^s}$-vector space.
		As in the previous example, all fractions of elements of~$U$ are in $\F_{q^m}$, and thus any nontrivial intersection $U\cap\alpha U$ arises from some $\alpha\in \F_{q^m}$.
		For any such~$\alpha$ we have $\dim (U+\alpha U)+\dim (U \cap \alpha U) = \dim U+\dim \alpha U$ and thus
		$\dim(U \cap\alpha U) = 2k - \dim(U+\alpha U)$.
		From $k \leq \dim(U+\alpha U) \leq m$ we obtain
		\[
                 2k-m\leq \dim_{\F_q} (U \cap \alpha U) \leq k.
        \]
      Since $\F_{q^s} \subset U$ is fixed by any shift by an element of $\F_{q^s}$  we see that there exist intersections  with $\dim(U \cap \alpha U) \geq s$.
      Because $U$ is not an $\F_{q^s}$-vector space it follows that these intersections are not all of $U$, hence $\ds(\Orb(U)) \leq 2(k-s)$.
       Choose now~$k$ such that $k\geq \frac{s+m}{2}$. Then the above shows that $\lambda_i=0$ for $i=1,\ldots,s-1$ as desired.
		\cref{ex:NonTrivStab2k4}(a) is an example of such a choice of $U$ with $q=3, s=2, m=4,n=8,$ and $k=3$.
	\end{example}
	
	In order to prove \cref{thm:Dist2k4L2Structure} we need to develop a few tools.
	The next definitions make sense even when $\ds(\Orb(U))<2k-4$, so we give the general versions before specializing to the scenario of \cref{thm:Dist2k4L2Structure}.
	
	\begin{definition}\label{def:A_V}
		Let $U\in \cG_q(k,n)$ such that $\ds(\Orb(U))=2k-2\ell$.
		For any subspace $V \subseteq U$ with $\dim(V)=\ell$, we define $A_V=\{\ol{\alpha}\in \bP(\F_{q^n}) \mid V \subseteq U \cap \alpha U \}$.
	\end{definition}

      Note that $A_V=\cS(U) \cup\{\ol{\alpha}\in \bP(\F_{q^n}) \mid V =U \cap \alpha U \}$, where $\cS(U)$ is as in \cref{def:UAssoc}.
      Then
      \begin{equation}\label{eq:AVcalS}
           \cS(U) \subsetneq A_V\Longleftrightarrow V=U\cap \alpha U\text{ for some }\alpha\in\F_{q^n}.
      \end{equation}
     We are, of course, interested in the case that $V$ arises as a maximal dimension intersection $U\cap \alpha U$ for some $\alpha$.
	To this end, we introduce a group action of $\F_q$ on $\F_{q^n}$.

    %%%%%%%%%%%%%%%%%%%%%%%%
	\begin{proposition}\label{prop:F_qGroupActionOnF_q^n}
		The map
        \[
              \varphi: (\F_{q^n} \sm \F_{q}) \times \F_q \longrightarrow (\F_{q^n}\sm\F_q),\quad (\alpha,\lambda)\longmapsto\frac{\alpha}{1+\lambda\alpha},
        \]
        is well-defined and satisfies the following properties.
		\begin{alphalist}
			\item The map $\varphi$ is a group action of $\F_q$ on $\F_{q^n}\sm\F_{q}$.
            \item Let $\F_{q^t}$ be a subfield of~$\F_{q^n}$ and $\lambda\in\F_q$. Then $\varphi(\alpha,\lambda)\in\F_{q^t}\sm\F_q\Longleftrightarrow\alpha\in\F_{q^t}\sm\F_q$.
			\item  $|\Orb_\varphi(\alpha)|=q$ for all $\alpha\in\F_{q^n}\sm\F_q$, and thus $\F_{q^t}\sm\F_q$ is the disjoint union of $q^{t-1}-1$ orbits for any divisor~$t$ of~$n$.
			\item For any $\alpha\in\F_{q^n}\sm\F_q$ the set $\ol{\Orb_\varphi(\alpha)}=\{\ol{\frac{\alpha}{1+\lambda\alpha}} \mid \lambda \in \F_q \}$ has cardinality $q$.
			\item Let $\alpha\in\F_{q^n}\sm\F_q$ and $\beta = \rho \alpha$ for some $\rho \in \F_q^*$. Then $\varphi(\beta,\lambda)=\rho\varphi(\alpha,\rho\lambda)$.
                     As a consequence, the set $\ol{\Orb_{\varphi}(\alpha)}$ depends only on the projective equivalence class $\ol{\alpha}$.
            \item $\bP(\F_{q^n})\sm\{\ol{1}\}$ is the disjoint union of $(q^{n-1}-1)/(q-1)$ sets of the form $\ol{\Orb_\varphi(\alpha)}$.
 		\end{alphalist}
	\end{proposition}
%%%%%%%%%%%%%%%%%%%%%

It should be noted that the map~$\varphi$ is not a well-defined map on equivalence classes in projective space, yet it leads to a disjoint union of projectivized orbits.

	\begin{proof}
		Since $\alpha \not \in \F_q$, we have $1+\lambda\alpha \neq 0$ for any $\lambda \in \F_q$.
		Further, suppose $\varphi(\alpha,\lambda)\in\F_q$, say
        $\frac{\alpha}{1+\lambda\alpha}=\mu \in \F_q$.
        Then $\alpha(1-\lambda\mu)=\mu$ and either $\alpha = \frac{\mu}{1-\mu\lambda}\in\F_q$ or $1-\lambda\mu=0$ and $\mu=0$.
		Since both are a contradiction,~$\varphi$ is well-defined.
		It remains to prove (a)--(f).
		\\
		(a) One straightforwardly verifies that $\varphi(\varphi(\alpha,\lambda),\mu)=\varphi(\alpha,\lambda+\mu)$.
		\\
        (b) Suppose $\varphi(\alpha,\lambda)\in\F_{q^t}\sm\F_q$, say $\frac{\alpha}{1+\lambda\alpha}=\mu \in \F_{q^t}\sm\F_q$.
        Then with the same reasoning as above we conclude $\alpha\in\F_{q^t}\sm\F_q$.
             The converse is trivial.
        \\
		(c) $\varphi(\alpha,\lambda)=\varphi(\alpha,\mu)$ implies $\lambda=\mu$ (since $\alpha\neq0$).
		Hence $|\Orb_{\varphi}(\alpha)|=q$, and the rest is clear.
		\\
		(d) We want to show that the cardinality of an orbit under $\varphi$ is preserved by passing to projective space.
		So suppose that $\ol{\varphi(\alpha,\lambda)}=\ol{\varphi(\alpha,\mu)}$.
		Then we have $1+\lambda\alpha=\rho(1+\mu\alpha)$ for some $\rho \in \F_q^*$.
		Since $\alpha \not \in \F_q$, $\{1,\alpha\}$ is $\F_q$-linearly independent and so we have $\rho =1$ and $\lambda=\mu$.
		\\
		(e) Suppose $\beta = \rho \alpha$ for some $\rho \in \F_q^*$.
		Then $\varphi(\beta,\lambda)=\frac{\beta}{1+\rho\lambda\alpha}=\rho\frac{\alpha}{1+\rho\lambda\alpha}=\rho\varphi(\alpha,\rho\lambda)$.
        \\
        (f)
        By~(d)  $\Orb_{\varphi}(\alpha)\cap\Orb_{\varphi}(\rho\alpha)=\emptyset$ for all $\rho\in\F_q^*\setminus\{1\}$ and $\alpha\in\F_{q^n}\sm\F_q$.
        On the other hand, thanks to~(e),
         $\ol{\Orb_{\varphi}(\alpha)}=\ol{\Orb_{\varphi}(\rho\alpha)}$ for all $\rho\in\F_q^*$.
         Since $|\bP(\F_{q^n})\setminus\{\ol{1}\}|=(q^n-q)/(q-1)$, all of this together with~(d) shows that
         the $q^{n-1}-1$ disjoint orbits covering $\F_{q^n}\sm\F_q$ collapse to $(q^{n-1}-1)/(q-1)$ projectivized orbits covering $\bP(\F_{q^n})\setminus\{\ol{1}\}$.
	\end{proof}

	We next show that the sets~$A_V$ decompose into the projectivized versions of the orbits of $\varphi$.
	
%%%%%%%%%%%%%%%%%%%%%%%%%%%%%
	\begin{proposition}\label{prop:A_VGroupAction}
		Let $U \in \cG_q(k,n)$ such that $\ds(\Orb(U))=2k-2\ell$ and let $\Stab(U)=\F_{q^t}^*$.
       Furthermore, let $V\subseteq U$  with $\dim(V)=\ell$. Then we have the following.
		\begin{alphalist}
			\item If $\ol{\alpha} \in A_V$, then $\ol{\varphi(\alpha,\lambda)} \in A_V$ for each $\lambda\in \F_q$.
            \item Suppose $V=U\cap\alpha U$ for some $\alpha\in\F_{q^n}$.
                     Write $A_V=\cS(U) \cup\hspace*{-.7em}\raisebox{.5ex}{$\cdot{}$}\hspace*{.45em} \cA$, where $\cS(U)$ is as in \cref{def:UAssoc}
                      and $\cA:=\{\ol{\beta}\mid V=U\cap \beta U\}$. Then $\cA$ is a disjoint union of projectivized orbits $\ol{\Orb_\varphi(\beta)}$.
			          Thus $|A_V| = aq + \frac{q^t-1}{q-1}$ for some $a\in\N$ and in particular $|A_V|\geq q+1$.
			\item If $\ol{\beta} = \ol{\varphi(\alpha,\lambda)} \in A_V$ for $\lambda \neq 0$, then $\ol{\beta\alpha^{-1}} = \ol{\varphi(\alpha^{-1},\lambda^{-1})} \in A_{\alpha^{-1}V}$.
		\end{alphalist}
	\end{proposition}
    %%%%%%%%%%%%%%%%%%%%%%%%%%

	\begin{proof}
		(a) Suppose $V\subseteq U\cap \alpha U$ and $\lambda \in \F_q$.
		We have to show that $V\subseteq U\cap \varphi(\alpha,\lambda)U$.
		Let $\{v_1,\ldots,v_m \}$ be a basis of $V$.
		Since $V \subset U \cap \alpha U$ there exist $\{u_1,\ldots,u_m\} \subset U$ such that $v_i=\alpha u_i$ for $i=1,\ldots,\ell$.
		Then $u_i+\lambda v_i \in U$ and
		\[ \varphi(\alpha,\lambda)(u_i+\lambda v_i)=\frac{\alpha u_i(1+\lambda \alpha)}{1+\lambda \alpha}=v_i. \]
		Therefore $V \subseteq U \cap \varphi(\alpha,\lambda) U$.
        \\
        (b) Let $V = U \cap \alpha U$ and $\lambda\in\F_q$. Then~(a) implies $V \subseteq U \cap \varphi(\alpha,\lambda) U$.
        Since~$\dim(V)=\ell$ is the maximum
         possible intersection dimension, we conclude that either $V= U \cap \varphi(\alpha,\lambda) U$ or $U \cap \varphi(\alpha,\lambda) U=U$.
         In the latter case $\varphi(\alpha,\lambda)\in\Stab(U)=\F_{q^t}^*$, and thus by Proposition~\ref{prop:F_qGroupActionOnF_q^n}(b) also $\alpha\in\F_{q^t}^*$, which is a contradiction.
         Thus $V= U \cap \varphi(\alpha,\lambda) U$.
         All of this shows that if $\ol{\alpha}\in\cA$, then $\ol{\varphi(\alpha,\lambda)}\in\cA$.
         With the aid of Proposition~\ref{prop:F_qGroupActionOnF_q^n}(f) we obtain that~$\cA$ is a disjoint union of projectivized orbits.
         The rest is clear.
         \\
		(c) Without loss of generality $\beta = \varphi(\alpha,\lambda)$. By assumption $V\subseteq U\cap \beta U$.
		Because $\alpha \in \Orb_{\varphi}(\beta)$, Part~(a) shows that $V\subseteq U\cap \alpha U$.
		Therefore $\alpha^{-1} V\subseteq U \cap \alpha^{-1} U$ and so $\ol{\alpha^{-1}} \in A_{\alpha^{-1}V}$.
		By~(a) again $\ol{\varphi(\alpha^{-1},\lambda^{-1})} \in A_{\alpha^{-1}V}$.
		Now
		\[ \beta\alpha^{-1}=\alpha^{-1}\varphi(\alpha,\lambda)=\frac{1}{1+\lambda\alpha}=\lambda^{-1}\frac{\alpha^{-1}}{1+\lambda^{-1}\alpha^{-1}}=\lambda^{-1}\varphi(\alpha^{-1},\lambda^{-1}). \]
		Since $\lambda^{-1}\in \F_q$, we conclude $\ol{\beta\alpha^{-1}}=\ol{\varphi(\alpha^{-1},\lambda^{-1})}$ and the claim follows.		
	\end{proof}

We now focus on subspaces~$U$ that generate a full-length orbit code (i.e., $\Stab(U)=\F_q^*$) and satisfy $\ds(\Orb(U))=2k-4$.
We will show next that in this case  $|A_V|=q+1$ for any $2$-dimensional intersection $V=U\cap\alpha U$; that is, $A_V$ is the union of $\{\ol{1}\}$ and a single projectivized orbit.
There are two possibilities for a $2$-dimensional subspace~$V$ over $\F_q$: either $V=\gamma\F_{q^2}$ for some $\gamma \in \F_{q^n}$ (hence~$n$ is even) or~$V$ itself has full-length orbit.
In the first case, we can explicitly describe $A_V$.
	
	%%%%%%%%%%%%%%%%%%%
	\begin{proposition}\label{prop:2k4IntReprNonTrivStab}
		Let $n$ be even and $U\in \cG_q(k,n)$ generate a full-length orbit with $\ds(\Orb(U))=2k-4$.
		Suppose there exist $\alpha,\gamma\in\F_{q^n}^*$ such that $V:=\gamma\F_{q^2}=U \cap \alpha U$.
	  Then $A_V=\bP(\F_{q^2})$ and thus $|A_V|=q+1$.
	\end{proposition}
	%%%%%%%%%%%%%%%%

	\begin{proof}
		We can reduce to the case $\gamma=1$ since any $\beta\in\F_{q^n}^*$ satisfies
		\begin{equation}\label{eq:Fq^2ShiftIntEquiv}
			\gamma\F_{q^2}=U\cap \beta U \Longleftrightarrow \F_{q^2}=\gamma^{-1} U\cap \beta\gamma^{-1} U
			= U'\cap \beta U',
		\end{equation}
		where $U'=\gamma^{-1}U$.
		Notice that $\Orb(U')=\Orb(U)$.
		If we can show that
			\[\{\ol{\beta} \mid \F_{q^2}\subseteq U'\cap \beta U' \} = \bP(\F_{q^2}), \]
		the claim follows from \eqref{eq:Fq^2ShiftIntEquiv}.
		Thus we assume that $V=\F_{q^2}=U\cap\alpha U$.
		
		In order to show $A_V\subseteq  \bP(\F_{q^2})$ suppose to the contrary that there exists $\beta \in A_V\sm \bP(\F_{q^2})$.
		Then $\F_{q^2}=U\cap\beta U$ and $\beta^{-1} \F_{q^2} \cap \F_{q^2}=\{0\}$.
		Furthermore,
		\[
		      \F_{q^2} \subset U\ \text{ and }\ \beta^{-1}\F_{q^2} \subset U.
		 \]
		Therefore $k\geq 4$ and we can decompose~$U$ as $U=\F_{q^2}\oplus \beta^{-1}\F_{q^2} \oplus U'$ for some $U' \in \cG_q(k-4,n)$.
		Now we have for any $\rho \in \F_{q^2}\sm\F_q$
		\[
		       \rho^{-1}(\F_{q^2}\oplus \beta^{-1}\F_{q^2})=\F_{q^2}\oplus \beta^{-1}\F_{q^2}\subseteq U,
		 \]
		so $\F_{q^2}\oplus \beta^{-1}\F_{q^2}\subseteq U\cap\rho U$ and thus $\dim(U\cap \rho U)\geq 4$.
		Since $\rho \not \in \Stab(U)=\F_q^*$, this  contradicts $\ds(\Orb(U))=2k-4$.
		All of this shows that $A_V \subseteq \bP(\F_{q^2})$.
		
		In the same way, since $\F_{q^2}\subseteq U$, every $\ol{\rho} \in \bP(\F_{q^2})$ leads to
		$\F_{q^2} \subseteq U \cap \rho U$, and thus
		$\ol{\rho}\in A_V$. Hence $\bP(\F_{q^2})\subseteq A_V$, and this concludes the proof.
	\end{proof}

\begin{rem}\label{rem:Vell}
The above result generalizes straightforwardly to full-length orbits with distance $2k-2\ell$ and intersections of the form
$V=\gamma\F_{q^\ell}=U\cap\alpha U$. In that case one arrives at $A_V=\bP(\F_{q^\ell})$.
\end{rem}
	
	The line of argument in the first part of the above proof can be extended to show that there is at most one such~$V$
	arising as an intersection $U\cap \alpha U$.
	
	\begin{proposition}\label{prop:NonTrivStab2dimIntCount2k4}
		Let  $U\in \cG_q(k,n)$ generate a full-length orbit with $\ds(\Orb(U))=2k-4$.
		Then there exists at most one subspace $V\in \cG_q(2,n)$ such that
		 $V=U\cap \alpha U$ for some $\alpha \in \F_{q^n}^*$ and
		$V=\gamma\F_{q^2}$ for some $\gamma \in \F_{q^n}^*$.
	\end{proposition}

	\begin{proof}
		If $n$ is odd, no such subspace exists, so let $n$ be even.
		Suppose to the contrary that there exist distinct subspaces $V_1,\,V_2$ such that
		$V_1=U\cap \alpha_1 U = \gamma_1 \F_{q^2}$ and $V_2=U\cap \alpha_2U =\gamma_2 \F_{q^2}$.
		We will show that $\ds(\Orb(U))<2k-4$.
		
		Since $V_1 \neq V_2$, we must have $\frac{\gamma_1}{\gamma_2}\not \in \F_{q^2}$ and even
		$\gamma_1\F_{q^2}\cap \gamma_2 \F_{q^2}=\{0\}$.
	    Now $\gamma_1\F_{q^2}\subseteq U$ and $\gamma_2\F_{q^2}\subseteq U$ implies
		$U=\gamma_1\F_{q^2}\oplus \gamma_2\F_{q^2}\oplus U'$ for some $U' \in \cG_q(k-4,n)$.
		This in turn leads to $\gamma_1\F_{q^2}\oplus \gamma_2\F_{q^2} \subseteq U \cap \rho U$ for any $\rho\in \F_{q^2} \sm\F_{q}$.
		Since $\rho\not\in\Stab(U)$, we arrive at the contradiction $\ds(\Orb(U))\leq 2k-8<2k-4$.
	\end{proof}
	
	It remains to describe the behavior when a maximal intersection $V = U \cap \alpha U$ has full-length orbit.
	In this case it turns out that there is a collection of related subspaces $\{V_i\}$ that all have associated sets $A_{V_i}$ of the same cardinality.
	Further, each $V_i$ is given by the cyclic shift $\alpha_i^{-1} V$ for some $\alpha_i \in A_{V}$ and we can explicitly describe the elements of $A_{V_i}$ in terms of those of $A_V$.
	This holds even in the more general setting where $\ds(\Orb(U))=2k-2\ell$.
	Although we will give the proof for the case $\ell=2$, the general case does not differ significantly.
	Therefore \cref{prop:MaxIntFullLenOrbitStructure} can be easily extended to the general case where $U$ has full-length orbit with $\ds(\Orb(U))=2k-2\ell$ and $V=U\cap \alpha U\in \cG_q(\ell,n)$
    such that $\Stab(V)=\F_q^*$.
	In particular, if $\gcd(\ell,n)=1$ then $\Stab(V)=\F_q^*$ for \emph{every} maximal intersection $V=U\cap \alpha U$ and the proposition applies.
	
	For the following result recall from \cref{def:A_V} that if $U$ generates a full-length orbit then
	$A_V=\{\ol{1}\}\cup\{\ol{\beta}\mid V=U\cap\beta U\}$.
	
	%%%%%%%%%%%%%%%%%%
	\begin{proposition}\label{prop:MaxIntFullLenOrbitStructure}
		Let $U \in \cG_q(k,n)$ generate a full-length orbit with $\ds(\Orb(U))=2k-4$.
		Furthermore, suppose there exists $V \in \cG_q(2,n)$ such that $V=U \cap \alpha_1 U$ for some $\alpha_1 \in \F_{q^n}^*$
		and~$V$ generates a full-length orbit.
		Let $|A_V|=s+1$ and write $A_V = \{\ol{\alpha_1},\ldots,\ol{\alpha_s},\ol{1} \}$.
		\begin{alphalist}
		\item The distinct cyclic shifts of $V$ that arise as intersections $U\cap \beta U$ are precisely the shifts by
			elements $\{\ol{\alpha_i^{-1}} \mid i=1,\ldots,s \} \cup \{\ol{1}\}$.
			In particular, there are $|A_V|=s+1$ such shifts.
		\item For each $i=1,\ldots,s$ we have
		        $A_{\alpha_i^{-1}V}=\{ \ol{\alpha_j\alpha_i^{-1}} \mid j=1,\ldots,s \} \cup \{ \ol{\alpha_i^{-1}} \}$ and
		        $|A_{\alpha_i^{-1}V}|=s+1$.
		\end{alphalist}
	\end{proposition}
	%%%%%%%%%%%%%%%%%
	
	\begin{proof}
	(a)	First we show that each $\alpha_i^{-1}V$ arises as an intersection $U\cap \beta U$.
		By assumption, $V=U\cap \alpha U = \langle v_0, w_0\rangle$ for some $v_0,w_0\in V$.
		Then for each $i \in \{1,\ldots,s\}$ there exist $v_i,w_i\in U$ such that
			\[v_0=\alpha_i v_i \quad \text{and}\quad w_0=\alpha_i w_i. \]
		Define $V_i=\alpha_i^{-1}V=\langle v_i, w_i \rangle$.
		Then
		\begin{equation}\label{eq:Vintersection}
		     U \cap \alpha_i^{-1} U= \alpha_i^{-1}(U\cap \alpha_i U) = \alpha_i^{-1}V=V_i,
		\end{equation}
	       so each of the shifts~$V_i$ arises as an intersection.
		
		Next we show that $V,V_1,\ldots,V_s$ are distinct. It follows immediately that $V\neq V_i$ for any~$i$ because
		$\Stab(V)=\F_q^*$ and $\ol{\alpha_i}\neq\ol{1}$.
		Suppose now $V_i=V_j$, thus $\langle v_i,w_i\rangle=\langle v_j,w_j\rangle$.
		From
		\begin{equation}\label{eq:vwi}
			v_i=\frac{\alpha_j}{\alpha_i}v_j \quad \text{and} \quad w_i=\frac{\alpha_j}{\alpha_i}w_j
	         \end{equation}
		we conclude that $\frac{\alpha_j}{\alpha_i} \in \Stab(V_i)=\Stab(V)=\F_q^*$, and since
		$\ol{\alpha_i}\neq \ol{\alpha_j}$ if $i\neq j$, we arrive at $i=j$.
		
		It remains to show that the shifts~$V_i$ are the only cyclic shifts of~$V$ that arise as intersections.
		Suppose $W=\gamma V=U \cap \beta U$ for some $\gamma,\beta \in \F_{q^n}^*$.
		Then $V=\gamma^{-1}W\subseteq\gamma^{-1}U$ and thus $V\subseteq U\cap\gamma^{-1}U$ by choice of~$V$.
		Thus $\ol{\gamma^{-1}}\in A_V$, as desired. 	
		\\		
		(b) Recall that $A_{V_i}=\{\ol{\beta} \in \bP(\F_{q^n}) \mid U\cap \beta U = V_i \}\cup\{\ol{1}\}$.
		We want to show that
		\begin{equation}\label{eq:AVi}
		      A_{V_i} =\left\{\ol{\frac{1}{\alpha_i}},\ol{\frac{\alpha_1}{\alpha_i}},\ldots,\ol{\frac{\alpha_s}{\alpha_i}} \right\}.
		 \end{equation}
		For ``$\supseteq$'' note that trivially $\ol{1}\in A_{V_i}$ and $\ol{\alpha_i^{-1}} \in A_{V_i}$
		thanks to~\eqref{eq:Vintersection}.
		Consider now $\ol{\alpha_j\alpha_i^{-1}}$ for $j\neq i$.
		Then $\ol{\alpha_j}\neq \ol{\alpha_i}$ and thus $\frac{\alpha_j}{\alpha_i} \not \in \F_q^*=\Stab(U)$.
		Moreover,~\eqref{eq:vwi} yields
		$V_i\subseteq U \cap \frac{\alpha_j}{\alpha_i}U$.
		Since by assumption the dimension of the intersection cannot be bigger than~$2$, we conclude
		$V_i=U\cap \frac{\alpha_j}{\alpha_i} U $.
                 \\
                	For ``$\subseteq$'' suppose $\ol{\beta} \in A_{V_i}\sm\{ \ol{\alpha_{i}^{-1}},\ol{1} \}$.
		Then there exist $u_1,u_2 \in U$ such that $v_i=\beta u_1$ and $w_i=\beta u_2$.
		Thus $\alpha_i\beta u_1=v_0$, $\alpha_i\beta u_2=w_0$ and $\ol{\alpha_i\beta} \neq \ol{1}$.
		All of this shows that $U \cap \alpha_i\beta U =V$.
		Hence $\ol{\alpha_i\beta} \in A_{V}$ and so $\ol{\alpha_i \beta} = \ol{\alpha_j}$ for some $j\in \{1,\ldots,s\}$.
		Thus $\ol{\beta} =\ol{\alpha_j\alpha_i^{-1}}$. This establishes~\eqref{eq:AVi}.
                 Finally, it is easy to see that the listed elements in~\eqref{eq:AVi} are distinct and thus $|A_{V_i}|=s+1$.		
		\end{proof}
	
	Our next result says that $|A_V|=q+1$ in the situation of \cref{prop:MaxIntFullLenOrbitStructure}.
	
	\begin{proposition}\label{prop:2k4IntRepCount}
		Let $U \in \cG_q(k,n)$ generate a full-length orbit with $\ds(\Orb(U))=2k-4$.
		Suppose there exists $V\in \cG_q(2,n)$ such that $V=U\cap \alpha U$ for some $\alpha \in \F_{q^n}^*$.
		Let $\ol{\beta} \in A_V \sm \{\ol{1}\}$.
		Then
		\begin{alphalist}
			\item $\ol{\alpha}=\ol{\beta}$ or $1=\lambda \alpha^{-1} + \mu \beta^{-1}$ for some $\lambda,\mu \in \F_q^*$,
			\item $|A_V|=q+1$.
		\end{alphalist}
	\end{proposition}
	
	\begin{proof}
		(a) Let $V=\langle v_1,v_2 \rangle = U \cap \alpha U$ and let $\ol{\beta} \in A_V\sm\{\ol{1}\}$.
		Then there exist $u_1,u_2,w_1,w_2 \in U$ such that
		\[
                    v_i=\alpha u_i = \beta w_i \text{ for }i=1,2.
        \]
		Note that $v_1,v_2 \in U$ as well.
		Define $\tilde{U}=\langle v_1, u_1, w_1 \rangle$. Then $\tilde{U} \subseteq U$ and
		\[
		     \frac{v_2}{v_1} \tilde{U}=\langle v_2, u_2, w_2 \rangle \subseteq U \cap \frac{v_2}{v_1} U.
		 \]
		Since $v_1,v_2$ are $\F_q$-linearly independent, the element $\frac{v_2}{v_1}$ is not in $\F_q^* = \Stab(U)$, and thus
		$U \cap \frac{v_2}{v_1} U\neq U$.
		Therefore ${\dim(U \cap \frac{v_2}{v_1} U) \leq 2}$ and so $\{v_1,u_1,w_1 \}$ must be $\F_q$-linearly dependent.
		
		Now we may argue as follows.
		First of all, the sets $\{v_1,u_1\}$ and $\{v_1,w_1\}$ are both $\F_q$-linearly independent because
		$\alpha, \beta \not \in \F_q$.
		Next, if $u_1=\lambda w_1$ for some $\lambda \in \F_q$ then
		$\beta w_1=\alpha u_1 =\alpha \lambda w_1$, and hence
		$\ol{\beta}=\ol{\alpha}$.
		It remains to consider the case $v_1=\lambda u_1+\mu w_1$ for some $\lambda,\mu \in \F_q^*$.
		But this implies immediately $1=\lambda \alpha^{-1} + \mu\beta^{-1}$, as desired.
		
		(b) Let $\ol{\beta}\in A_V\setminus\{\ol{1},\ol{\alpha}\}$. Part~(a) tells us that
		 $\mu\beta^{-1} = 1-\lambda\alpha^{-1}$ for some $\lambda,\mu\in\F_q^*$.
		 The $q-1$ choices for~$\lambda$ imply that there are at most $q-1$ options for such $\ol{\beta}$.
		 Thus $|A_V|\leq q+1$.
         The reverse inequality has been established in \cref{prop:A_VGroupAction}(b).
	\end{proof}

	We now have all of the pieces in place to prove \cref{thm:Dist2k4L2Structure}.
	
	\begin{proof}[Proof of {\cref{thm:Dist2k4L2Structure}}]
		By applying \eqref{eq:t1sums} and \cref{def:DistDistr}  with $t=1$ and $\ell=2$ we notice that it suffices to compute $\lambda_2$ for such a subspace $U$.
        Hence we need to determine $|\{\ol{\beta}\mid \dim(U\cap\beta U)=2\}|$.
        We distinguish two cases.	
       \\
       \underline{Case 1:}  Suppose $U$ contains a cyclic shift of $\F_{q^2}$.
       Then $U=\gamma\F_{q^2}\oplus U'$ for some $U' \in \cG_q(k-2,n)$ and $\gamma \in \F_{q^n}$.
       Then $|A_{\gamma\F_{q^2}}\sm\{\ol{1}\}|=q$ by \cref{prop:2k4IntReprNonTrivStab}, and this is the number of $\ol{\beta}\in\bP(\F_{q^n})$ such that $U\cap\beta U=\gamma\F_{q^2}$.
        Moreover, thanks to \cref{prop:NonTrivStab2dimIntCount2k4} there does not exist any $V \in \cG_q(2,n)$ such that $V= \gamma'\F_{q^2}=U \cap \alpha U$ for some
        $\alpha,\gamma' \in \F_{q^n}$ except for $V=\gamma\F_{q^2}$.
        In other words, any other $2$-dimensional intersection $V=U\cap \alpha U$ has full-length orbit.
		\cref{prop:2k4IntRepCount} shows that for each such~$V$ we have $|A_V\sm\{\ol{1}\}|=q$, which is the number of $\ol{\beta}\in\bP(\F_{q^n})$ leading to the $2$-dimensional intersection~$V$.
       Furthermore, \cref{prop:MaxIntFullLenOrbitStructure} says that the collection of all $2$-dimensional intersections with full-length orbit can be partitioned into sets of the form
       $\{V,\alpha_1^{-1}V,\ldots,\alpha_q^{-1}V\}$, each one with cardinality $q+1$. Suppose we have~$r$ such sets. Note that $r=0$ is possible.
       Then
		\[
                     \sum_{\substack{V=U\cap\alpha U \\ \substack{\dim(V)=2 \\ V\neq \gamma\F_{q^2}} }}|A_V\sm\{\ol{1}\}|=rq(q+1).
         \]
		Combining all of this, we arrive at $\lambda_2=q+rq(q+1)$, as desired.
		\\
		\underline{Case 2:}  Suppose $U$ does not contain a cyclic shift of $\F_{q^2}$.
       Then any $V\in \cG_q(2,n)$ with $V=U\cap\alpha U$ for some $\alpha$ has full-length orbit and since $\ds(\Orb(U))=2k-4$ there exists at least one such subspace.
		So the previous argument shows that
		\[
              \sum_{\substack{V=U\cap\alpha U \\\dim(V)=2 }}|A_V\sm\{\ol{1}\}|=rq(q+1)\text{ for some }r\geq1
         \]
		and $\lambda_2=rq(q+1)$, as stated.
	\end{proof}

We conclude this section with some examples illustrating various intersection distributions for full-length orbits with distance $2k-4$.
     We used SageMath to compute the values of $\lambda_2$ and $r$ that occurred for some different values of the parameters $q,n$, and $k$.
     Recall from \cref{thm:Dist2k4L2Structure} that~$\lambda_2$ fully determines the intersection distribution.
     For each triple $(q,n,k)$, we generated random subspaces in $\cG_q(k,n)$ containing the element~$1$ and analyzed those that
     generated a full-length orbit with distance $2k-4$.
     In the following table we list all occurring values for~$\lambda_2$ along with their frequency~$N$, ordered accordingly.
     For example, when $(q,n,k)=(2,10,4)$ we found 248 subspaces with $\lambda_2=2$ and $2598$ subspaces with
     $\lambda_2=6$ etc.
     In the same way we list the corresponding value of~$r$.
     Recall from \cref{cor:2k-4DistIneq} that the maximum possible value for $\lambda_2$  is $Q/(q+1)$.

     \begin{table}[h!]
     	\centering
       \footnotesize{
     	\begin{tabular}{c|c|c|c|c|c|c}
     		$q$ \!\!&\!\! $n$ \!\!&\!\! $k$ \!\!&\!\!$\lambda_2$ \!\!&\!\! $r$ \!\!&\!\! $N$ \!\!&\!\! $\frac{Q}{q+1}$\\
     		\hline $2$ \!\!&\!\! $10$ \!\!&\!\! $4$ \!\!&\!\! $2,6,8,12,14,$ \!\!&\!\! $0,1,1,2,2,$ \!\!&\!\! $248,2598,34,2059,90,$ \!\!&\!\! $70$ \\
     		\!\!&\!\! \!\!&\!\! \!\!&\!\! $18,20,24,30$ \!\!&\!\! $3,3,4,5$ \!\!&\!\! $298,94,195,49$ \!\!&\!\! \\
     		\hline $2$ \!\!&\!\! $11$ \!\!&\!\! $4$ \!\!&\!\! $6,12,18,24,30,42$ \!\!&\!\! $1,2,3,4,5,7$ \!\!&\!\! $1760,1251,63,57,12,24 $\!\!&\!\! $70$  \\
     		\hline $2$ \!\!&\!\! $11$ \!\!&\!\! $5$ \!\!&\!\! $6,12,18,24,30,36,42,48,$ \!\!&\!\! $1,2,3,4,5,6,7,8,$ \!\!&\!\! $3,7,22,67,243,494,982,1228,$ \!\!&\!\! $310$ \\
     		\!\!&\!\! \!\!&\!\! \!\!&\!\! $54,60,66,72,78,84,90,$ \!\!&\!\! $9,10,11,12,13,14,15,$ \!\!&\!\! $1483,1285,1143,783,519,258,153,$ \!\!&\!\!\\
     		\!\!&\!\! \!\!&\!\! \!\!&\!\! $96,102,108,114,120,126$ \!\!&\!\! $16,17,18,19,20,21$ \!\!&\!\! $90,39,29,11,2,1 $ \!\!&\!\!\\
     		\hline $2$ \!\!&\!\! $12$ \!\!&\!\! $4$ \!\!&\!\! $2,6,8,12,14,18,$ \!\!&\!\! $0,1,1,2,2,3,$ \!\!&\!\! $150,953,4,664,6,13,$ \!\!&\!\! $70$\\
     		\!\!&\!\! \!\!&\!\! \!\!&\!\! $20,24,30,42$ \!\!&\!\! $3,4,5,6$ \!\!&\!\! $9,13,2,4$ \!\!&\!\! \\
     		\hline $2$ \!\!&\!\! $13$ \!\!&\!\! $4$ \!\!&\!\! $6,12,18,24,30,42$ \!\!&\!\! $1,2,3,4,5,7$ \!\!&\!\! $ 1486, 967, 7, 8, 3,18 $ \!\!&\!\! $70$\\
     		\hline $2$ \!\!&\!\! $13$ \!\!&\!\! $5$ \!\!&\!\! $6,12,18,24,36,42,$ \!\!&\!\! $1,2,3,4,5,6,7,$ \!\!&\!\! $ 1136,2933,1535,1485,528,$ \!\!&\!\! $310$\\
     		\!\!&\!\! \!\!&\!\! \!\!&\!\! $48,54,60,66,72$ \!\!&\!\! $8,9,10,11,12$ \!\!&\!\! $437,148,97,36,26,6,6 $ \!\!&\!\!\\
     		\hline $3$ \!\!&\!\! $9$ \!\!&\!\! $4$ \!\!&\!\! $12,24,36,48,60,72,84$ \!\!&\!\! $1,2,3,4,5,6,7$ \!\!&\!\! $2900, 3537, 283, 354, 290, 160, 55$ \!\!&\!\! $390$\\
     		\hline $3$ \!\!&\!\! $11$ \!\!&\!\! $4$ \!\!&\!\! $12,24,36,48,60,72,84$ \!\!&\!\! $1,2,3,4,5,6,7 $ \!\!&\!\! $1048,1091,4,8,12,6,2 $ \!\!&\!\! $390$\\
     		\hline $3$ \!\!&\!\! $12$ \!\!&\!\! $4$ \!\!&\!\! $3,12,24,27,36,48$ \!\!&\!\! $0,1,2,2,3,4$ \!\!&\!\! $ 21, 288, 397,1,2,4 $ \!\!&\!\! $390$\\
     	\end{tabular}
        }
     	\caption{Example values of $\lambda_2,r$ for random search of full-length orbits with distance $2k-4$}
     	\label{tab:rValues}
     \end{table}

	Besides the random searches that we present above, we also performed, for various choices of parameters, exhaustive searches among all subspaces in $\cG_q(k,n)$ that contain $1$.
	We mostly restricted ourselves to $k=3$ because of computational feasibility.
	The results of these exhaustive searches are presented in \cref{tab:rValuesExhaustive} on the next page.
	Again, the values of $N$ in the table are the frequencies of the values of $\lambda_2$ (and $r$).
	Notice in particular that the values of $\lambda_2$ we found by random search in \cref{tab:rValues} for $q=2,n=10,k=4$ do indeed exhaust all possible values of $\lambda_2$ for these parameters.
	
	Each value of $N$ appearing in \cref{tab:rValuesExhaustive} is a multiple of $(q^k-1)/(q-1)$; this is due to the fact that for any subspace $U$ containing $1$, the shifts $\alpha^{-1}U$ for
	$\alpha \in U\sm 0$ also contain~$1$ and generate the same orbit as $U$.  Furthermore, these are all of the elements of $\Orb(U)$ that contain~$1$.
	This means that our exhaustive search counts every cyclic orbit code $(q^k-1)/(q-1)$ times.
	Note that \cref{tab:rValuesExhaustive} shows that the upper bound for $\lambda_2$ in \cref{cor:2k-4DistIneq} is quite poor in general.
	
	\begin{table}[h!]
		\centering
		\footnotesize{
		\begin{tabular}{c|c|c|c|c|c|c}
			$q$ \!\!&\!\! $n$ \!\!&\!\! $k$ \!\!&\!\! $\lambda_2$ \!\!&\!\! $r$ \!\!&\!\! $N$ \!\!&\!\! $\frac{Q}{q+1}$ \\
			\hline $2$ \!\!&\!\! $6$ \!\!&\!\! $3$ \!\!&\!\! $2,6$ \!\!&\!\! $0,1$ \!\!&\!\! $35, 63$ \!\!&\!\! $14$ \\
			\hline $2$ \!\!&\!\! $7$ \!\!&\!\! $3$ \!\!&\!\! $6$ \!\!&\!\! $1$ \!\!&\!\! $147$ \!\!&\!\! $14$ \\
			\hline $2$ \!\!&\!\! $8$ \!\!&\!\! $3$ \!\!&\!\! $2,6,14$ \!\!&\!\! $0,1,2$ \!\!&\!\! $140,280,7 $ \!\!&\!\! $14$ \\
			\hline $2$ \!\!&\!\! $8$ \!\!&\!\! $4$ \!\!&\!\! $12,14,18,20$ \!\!&\!\! $2,2,3,4,$ \!\!&\!\! $1080,1200,3000,1200,$ \!\!&\!\! $70$ \\	
			\!\!&\!\! \!\!&\!\! \!\!&\!\! $24,30,38$ \!\!&\!\! $4,5,6$ \!\!&\!\! $2760, 1200, 750$ \!\!&\!\! \\	
			\hline $2$ \!\!&\!\! $9$ \!\!&\!\! $3$ \!\!&\!\! $6$ \!\!&\!\! $1$ \!\!&\!\! $588$ \!\!&\!\! $14$ \\
			\hline $2$ \!\!&\!\! $9$ \!\!&\!\! $4$ \!\!&\!\! $6,12,18,24,30$ \!\!&\!\! $1,2,3,4,5$ \!\!&\!\! $31995, 33120, 11340, 7560, 2025$ \!\!&\!\! $70$ \\
			\hline $2$ \!\!&\!\! $10$ \!\!&\!\! $3$ \!\!&\!\! $2,6$ \!\!&\!\! $0,1$ \!\!&\!\! $595,1190 $ \!\!&\!\! $14$ \\
			\hline $2$ \!\!&\!\! $10$ \!\!&\!\! $4$ \!\!&\!\! $2,6,8,12,14,$ \!\!&\!\! $0,1,1,2,2,$ \!\!&\!\! $35700,213375,2550,164235,7650,$ \!\!&\!\! $70$ \\	
			\!\!&\!\! \!\!&\!\! \!\!&\!\! $18,20,24,30$ \!\!&\!\! $3,3,4,5$ \!\!&\!\! $22725, 7650, 14325,3750$ \!\!&\!\! \\	
			\hline $3$ \!\!&\!\! $6$ \!\!&\!\! $3$ \!\!&\!\! $3,12$ \!\!&\!\! $0,1$ \!\!&\!\! $130,377$ \!\!&\!\! $39$ \\	
			\hline $3$ \!\!&\!\! $7$ \!\!&\!\! $3$ \!\!&\!\! $12$ \!\!&\!\! $1$ \!\!&\!\! $1183$ \!\!&\!\! $39$ \\	
			\hline $3$ \!\!&\!\! $8$ \!\!&\!\! $3$ \!\!&\!\! $3,12,39$ \!\!&\!\! $0,1,3$ \!\!&\!\! $1170,3510,13 $ \!\!&\!\! $39$ \\	
			\hline $3$ \!\!&\!\! $9$ \!\!&\!\! $3$ \!\!&\!\! $12$ \!\!&\!\! $1$ \!\!&\!\! $10647$ \!\!&\!\! $39$ \\
			\hline $5$ \!\!&\!\! $6$ \!\!&\!\! $3$ \!\!&\!\! $5,30$ \!\!&\!\! $0,1$ \!\!&\!\! $806,3999$ \!\!&\!\! $155$ \\	
		\end{tabular}
		}
	\caption{Values of $\lambda_2,r$ for exhaustive search of full-length orbits with distance $2k-4$}
	\label{tab:rValuesExhaustive}
	\end{table}
	
	Finally we present an example concerning the value $f(U)$.
	From \cref{theo:SidonCount} and \cref{cor:2k-4DistIneq} we know that for full-length orbits $\Orb(U)$
	with distance $2k-2$ or $2k-4$ the intersection distribution is completely determined by $q,n,k,$ and $f(U)$; see
	also~\eqref{eq:flambda}.
	In fact, this also holds when the distance is $2k$ because in that case the intersection distribution is trivial.
	It is thus natural to ask whether $q,n,k,f(U)$ along with the distance also determine the intersection distribution of
	 full-length orbits if the distance is at most $2k-6$.
	
	However, this does not hold. Furthermore, $q,n,k,$ and $f(U)$ do not determine the distance of the orbit code.
	
	\begin{example}\label{ex:fDoesNotDetermineDistrib}
		Let $q=2, n=11, k=5$ and suppose $\omega$ is a primitive element of $\F_{2^{11}}$ over $\F_2$ satisfying $\omega^{11}=\omega^2+1$. Define
			\[
				U\!=\! \langle 1, \omega^{417},\omega^{1823},\omega^{1983},\omega^{64}\rangle,\
				V\!=\! \langle 1, \omega^{1332},\omega^{468},\omega^{749},\omega^{1627}\rangle, \
				W\!=\!\langle 1, \omega^{1618}, \omega^{942}, \omega^{1041}, \omega^{1315}\rangle.
			\]
		Then all three subspaces generate full-length orbits. A computation using SageMath shows that
		$\ds(\Orb(U))=\ds(\Orb(V))=4=2k-6$, whereas $\ds(\Orb(W))=6=2k-4$.
		Furthermore,  $f(U)=f(V)=f(W)=703$, and hence $f$ does not determine the distance.
		Finally $\Orb(U)$ has intersection distribution $(\lambda_0,\lambda_1,\lambda_2,\lambda_3)=(1343,624,60,18)$
		while $\Orb(V)$ has $(\lambda_0,\lambda_1,\lambda_2,\lambda_3)=(1343,600,96,6)$.
		Thus, the intersection distribution is not determined by $q,n,k,f$ and the distance.
	\end{example}
	
	%%%%%%%%%%%%%%%%%%%%%%%%%%%%%%%%%%%%%%%%%%%%%%%%%%%%%%
	\section{Intersection Distributions of Unions of Full-Length Orbits} \label{sec:UnionOfOrbits}
	%%%%%%%%%%%%%%%%%%%%%%%%%%%%%%%%%%%%%%%%%%%%%%%%%%%%%%
	
	In this section we generalize  the ideas of Sections~\ref{sec:Preliminaries} and \ref{sec:fractions} to codes that arise as union of orbits generated by subspaces of the same dimension.
	We need to start by adapting the definitions from the single orbit case to multiple orbits.
	Analogously to \cref{def:DistDistr} we define the distance and intersection distributions with respect fixed reference spaces for each orbit.
	In order to relate these two distributions we need to restrict ourselves to subspaces with the same stabilizer.
	
	%%%%%%%%%%%%%%%%%%%%%%%%%%%%
	\begin{definition}\label{def:DistDistrMultOrbits}
		Let $k\leq n/2$ and $U_j \in\Gkn$ with $\Stab(U_j)=\F_{q^t}$ for $j=1,\ldots,m$ and define $\cC=\bigcup_{j=1}^m\Orb(U_j)$. Suppose $\ds(\cC)=2d$.
		For $i=d,\ldots,k$ define $\delta_{2i}=|\{(U_j, \alpha U_{j'}) \mid 1\leq j \leq j' \leq m,\; \d(U_j,\alpha U_{j'})=2i\}|$.
		We call $(\delta_{2d},\ldots,\delta_{2k})$ the \emph{distance distribution} of $\cC$.
		Furthermore, we set $\ell=k-d$, thus $\ell$ is the maximum dimension of the intersection spaces $U_j\cap\alpha U_{j'}$.
		For $i=0,\ldots,\ell$ and $1\leq j \leq j' \leq m$ we define
		$
		\cL_i(U_j,U_{j'})=\{\ol{\alpha}\in\bP(\F_{q^n})\mid \dim(U_j\cap\alpha U_{j'})=i\}
		$
		and set $\lambda_i$ as
		\[
		\lambda_i(\cC) =  \sum_{j \leq j'} |\cL_{i}(U_j,U_{j'})|.
		\]
		We call
		$(\lambda_0,\ldots,\lambda_\ell)$ the \emph{intersection distribution} of $\cC$.
		As in \cref{def:DistDistr} we have $\lambda_i=(q^t-1)/(q-1)\delta_{2(k-i)}$ for $i=0,\ldots,\ell$.
	\end{definition}
	%%%%%%%%%%%%%%%%%%%%%%%%%%%%%
	
    We now extend \cref{def:UAssoc} to the case of multiple generating subspaces.
	It will suffice to extend the definitions to pairs $(U,V)$ of subspaces and we will do so for $\cL,\,\cM$, and $\cF$.
	There is no meaningful generalization of $\cS(U)$ to two spaces, and in fact no such space will be needed.
	
	%%%%%%%%%%%%%%%%%%%%%%%%%%%%%%%%%
	\begin{definition}\label{def:MultOrbUAssoc}
		Let $k\leq n/2$ and $U,V \in\Gkn$. Define $\ell=\max\{\dim(U\cap \alpha V) \mid \alpha \in \F_{q^n}^* \}$.
		Note that the cyclic code $\Orb(U)\cup \Orb(V)$ has minimum distance at most $2k-2\ell$.
		We define
		\begin{arabiclist}
			\item $\cL(U,V)=\bigcup_{i=1}^\ell\cL_i(U,V)$, where $\cL_i(U,V)$ is as in \cref{def:DistDistrMultOrbits}.
			\item $\cM(U,V)=\{(\ol{u},\ol{v})\mid u\in U\setminus 0, v \in V \setminus0\}$.
			\item $\cF(U,V)=\{ \ol{uv^{-1}} \mid u \in U \sm 0, v\in V\sm 0 \}$ and $f_{U,V}:=|\cF(U,V)|$.
		\end{arabiclist}
	\end{definition}
	
	Notice that when $U=V$, each of these definitions reduces to the corresponding part in \cref{def:UAssoc}.
	As in the single subspace case, we omit $i=0$ from the definition of $\cL(U,V)$ since $\lambda_0$ can be calculated from $\lambda_i,\,i=1,\ldots,\ell$.
	We will carry this out in the proof of \cref{thm:DistDistribUnionTwoSpaceSidon}.
	
	Again, the cardinality of $\cM(U,V)$ depends only on $q,n,$ and $k$ and we denote this by
	\begin{equation}\label{eq:QMultOrbits}
	\widehat{Q}:= |\cM(U,V)| = \left(\frac{q^k-1}{q-1} \right)^2.
	\end{equation}
	
	For any two subspaces $U,V \in \Gkn$ generating different orbit codes we have again a map $\psi: \cM(U,V) \to \bP(\F_{q^n})$ given by $\psi(\ol{u},\ol{v})=\ol{uv^{-1}}$.
	As in \cref{sec:fractions}, $\psi$ surjects onto $\cL(U,V)$.
	
	\begin{proposition}\label{prop:psiPreimageSizeMultOrbits}
		Let $U,V\in\Gkn$ such that $\Orb(U) \neq \Orb(V)$ and set $\ell=\max\{\dim(U\cap \alpha V) \mid \alpha \in \F_{q^n}^* \}$.
		The map
		$\psi: \cM(U,V) \to\cF(U,V),\quad (\ol{u},\ol{v})\longmapsto \ol{uv^{-1}}$ is well-defined.
		It satisfies $\cF(U,V)=\psi(\cM(U,V))=\cL(U,V)$.
		Furthermore,  for any $\ol{\alpha}\in\cF(U,V)$ we have
		\[
		\ol{\alpha}\in\cL_i(U,V)\Longleftrightarrow |\psi^{-1}(\ol{\alpha})|=(q^i-1)/(q-1).
		\]
	\end{proposition}
	
	\begin{proof}
		The well-definedness of $\psi$ is clear and so is $\psi(\cM(U,V))=\cF(U,V)$.
		We show next that $\psi(\cM(U,V))=\cL(U,V)$.
		First, let $\ol{\alpha}=\ol{uv^{-1}}\in \psi(\cM(U,V))$ for some $u\in U\setminus 0, v \in V \setminus0$.
		Then there exists $\lambda\in\F_q^*$ such that $u=\lambda\alpha v$.
		Since $\lambda V=V$ this implies $U\cap \alpha V\neq 0$.
		Hence $\dim(U\cap\alpha V)\in\{1,\ldots,\ell\}$ and thus $\ol{\alpha}\in \cL(U,V)$.
		This shows $\psi(\cM(U,V)) \subseteq \cL(U,V)$.
		The proof of the reverse inclusion proceeds similarly.
		If $\ol{\alpha}\in \cL(U,V)$, then $1\leq\dim(U\cap\alpha V)\leq \ell$.
		Hence there exist $u\in U\sm0, v\in V\sm0$ with $u=\alpha v$.
		So $\ol{\alpha}=\ol{uv^{-1}}$ is in $\psi(\cM(U,V))$.
		
		It remains to show $\ol{\alpha}\in\cL_i(U,V)\Longleftrightarrow |\psi^{-1}(\ol{\alpha})|=(q^i-1)/(q-1)$.
		Fix $\ol{\alpha}\in\cF(U,V)$.
		Note first that for $(\ol{v},\ol{w})\in\psi^{-1}(\ol{\alpha})$, the second component $\ol{w}$ is uniquely determined by the first one.
		Thus it suffices to count the number of possible first components.
		
		Let $\ol{\alpha}\in\cL_i(U,V)$. Hence $\dim(U\cap \alpha V)=i$.
		Then for every $v\in U\cap\alpha V$ there exists $w\in V$ such that $v=\alpha w$.
		Using that there exist $(q^i-1)/(q-1)$ elements $\ol{v}$ such that $v\in U\cap\alpha V$, it follows that
		$|\psi^{-1}(\ol{\alpha})|\geq(q^i-1)/(q-1)$.
		Conversely, suppose that $(\ol{x},\ol{y}) \in \psi^{-1}(\ol{\alpha})$.
		Then $\ol{x}=\ol{\alpha y}$ and $x=\lambda\alpha y$ for some $\lambda\in\F_q$.
		Thus $x\in U\cap\alpha V$.
		This leaves $(q^i-1)/(q-1)$ choices for $\ol{x}$, and thus $|\psi^{-1}(\ol{\alpha})|\leq \frac{q^i-1}{q-1}$.
		Hence $|\psi^{-1}(\ol{\alpha})|= \frac{q^i-1}{q-1}$.
	\end{proof}
	
	The restriction that $\Orb(U)\neq \Orb(V)$ implies that $U\neq\alpha V$ for any~$\alpha$,
	hence the differences between the statements of \cref{prop:psiPreimageSize} and \cref{prop:psiPreimageSizeMultOrbits}.
	As in \cref{sec:fractions}, we can use this result to derive identities relating the sizes $|\cL_i(U,V)|$ for $i=1,\ldots,\ell$.
	
	%%%%%%%%%%%%%%%%%
	\begin{corollary}\label{cor:CountMultOrbits}
		Let $U,V \in\Gkn$ such that $\Orb(U) \neq \Orb(V)$. Let $\ell=\max\{\dim(U\cap \alpha V) \mid \alpha \in \F_{q^n}^* \}$.
		Recall the cardinalities $f_{U,V}=|\cF(U,V)|,\,\widehat{Q}=|\cM(U,V)|$, and set $\lambda_i=|\cL_i(U,V)|$ for $i=1,\ldots,\ell$.
		Then
		\begin{equation}\label{eq:fsumMultOrbits}
		f_{U,V}= \sum_{i=1}^\ell \lambda_i.
		\end{equation}
		and
		\begin{equation}\label{eq:QsumMultOrbits}
		\widehat{Q}=\sum_{i=1}^{\ell}\frac{q^i-1}{q-1}\lambda_{i}.
		\end{equation}
	\end{corollary}
	%%%%%%%%%%%%%%
	
	\begin{proof}
		The identity in \eqref{eq:fsumMultOrbits} follows immediately from $\cL(U,V)=\cF(U,V)$  from \cref{prop:psiPreimageSizeMultOrbits}.
		From the same proposition we have $\psi(\cM(U,V))=\cL(U,V)$, thus
		$\cM(U,V)=\bigcup_{i=1}^\ell\psi^{-1}(\cL_i(U,V))$.
		Now \eqref{eq:QsumMultOrbits} follows from \cref{prop:psiPreimageSizeMultOrbits} and the cardinality of $\cM(U,V)$ in \eqref{eq:QMultOrbits}.
	\end{proof}
	
	In the single orbit case, we saw that orbit codes generated by a Sidon space have full length and maximal possible dimension.
	The Sidon property can be extended to various spaces in such a way that the orbits stay sufficiently far away from each other in the subspace distance.
	This reads as follows.
	
	\begin{definition}\label{def:TwoSpaceSidon}
		Let $U$ and $V$ be distinct $k$-dimensional subspaces of $\F_{q^n}$.
		We say that $U$ and $V$ have the \emph{two-space Sidon property} if any $a,c \in U\sm0$ and $b,d \in V \sm0$ with $ab=cd$ satisfy $\ol{a}=\ol{c}$ and $\ol{b}=\ol{d}$.
	\end{definition}

	\begin{lemma}[{\cite[Lemma 36]{RRT18}}]\label{lem:2spaceSidon}
		Let $U$ and $V$ be distinct subspaces in $\cG_q(k,n)$.
		The following conditions are equivalent.
		\begin{alphalist}
			\item $\dim(U\cap\alpha V) \leq 1$ for all $\alpha \in \F_{q^n}^*$.
			\item $U$ and $V$ have the two-space Sidon property.
		\end{alphalist}
		As a consequence, if~$U,V\in\cG_q(k,n)$ are Sidon spaces and have the two-space Sidon property then the cyclic subspace code $\Orb(U)\cup\Orb(V)$
		has cardinality $2(q^n-1)/(q-1)$ and distance $2k-2$.
	\end{lemma}
	
	Comparing the above definition with \cref{def:Sidon} one may wonder whether the assumption $ab=cd$ should also allow for the conclusion $\ol{a}=\ol{d}$ and $\ol{c}=\ol{b}$,
	which is an option in the case where $a,b,v,d\in U\cap V$.
	However, this is not necessary.
	The difference between the above lemma and the situation in Theorem~\ref{thm:Sidon}  lies in the obvious fact that the property $\dim(U\cap\alpha V)\leq 1$ for all $\alpha\in\F_{q^n}^*$
	can never be true if $V=U$.
	For further details we refer to the proofs of \cref{thm:Sidon} and \cref{lem:2spaceSidon} in \cite{RRT18}.

	We can use the above lemma to extend our earlier results to cyclic codes that are unions of cyclic orbits generated by Sidon spaces that pairwise have the two-space Sidon property.
	Such codes have maximal possible length and distance $2k-2$.
	Their existence has been established in \cite[Construction 37]{RRT18}, where the authors give a construction as a generalization of their own \cite[Construction 15]{RRT18}.
	Another construction from the same paper, \cite[Construction 11]{RRT18}, can be generalized in the same way to give another class of cyclic orbit codes of this type.
	
	We have now all of the necessary pieces to describe the intersection distribution of such codes.
	
	%%%%%%%%%%%%%%%%%%%%%%%%%%%%
	\begin{theorem}\label{thm:DistDistribUnionTwoSpaceSidon}
		Let $U_1,\ldots,U_m \in \cG_q(k,n)$ be distinct subspaces such that each $U_i$ is a Sidon space and each pair $U_i, U_j$ with $i\neq j$ has the two-space Sidon property.
		Let $\cC = \bigcup_{i=1}^m \Orb(U_i).$
		Then $\cC$ is a cyclic subspace code with $|\cC|= m(q^n-1)(q-1)^{-1}$ and $\ds(\cC)=2k-2$.
		Further $\cC$ has intersection distribution $(\lambda_0,\lambda_1)$ where	
		\begin{align*}
		\lambda_1&=m\left( \frac{q^k-1}{q-1}\right) \left( \frac{q^k-q}{q-1}\right)  +\binom{m}{2}\left(\frac{q^k-1}{q-1}\right)^2,\\
		\lambda_0&=m\left(\frac{q^n-1}{q-1}-Q-1\right)+\binom{m}{2} \left(\frac{q^n-1}{q-1}-\widehat{Q}\right).
		\end{align*}
	\end{theorem}
	%%%%%%%%%%%%%%%%%%%%%%%%%%%%%%%%
	
	\begin{proof}
		Clearly $\cC$ is a cyclic subspace code (in the sense of the paragraph before Definition~\ref{def:COC}).
		Each orbit has size $(q^n-1)(q-1)^{-1}$ since it is generated by a Sidon space, and because each pair has the two-space Sidon property, \cref{lem:2spaceSidon} implies that the orbits are disjoint.
		Hence $|\cC|= m(q^n-1)(q-1)^{-1}$.
		As for the minimum distance note that on the one hand $\min\{\d(U_j, \alpha U_{j'}) \mid j<j',\,\alpha \in \F_{q^n}^* \} \geq 2k-2$ thanks to
		\cref{lem:2spaceSidon}  while on the other hand each orbit itself has distance $2k-2$ by \cref{thm:Sidon}.
		
		For the intersection distribution define $\lambda_{i,j,j'}=|\cL_i(U_j,U_{j'})|$ for $i=0,1$ and $1\leq j\leq j'\leq m$.
		Recall that $\lambda_i=\sum_{j\leq j'}\lambda_{i,j,j'}$.
		Since each $U_j$ is a Sidon space, \cref{theo:SidonCount} gives us
		\[
		\lambda_{1,j,j}=Q=\left( \frac{q^k-1}{q-1}\right)\left( \frac{q^k-q}{q-1}\right).
		\]
		For $j<j'$, \cref{eq:QsumMultOrbits} gives
		\[
		\lambda_{1,j,j'}=\widehat{Q}=\left( \frac{q^k-1}{q-1} \right)^2.
		\]
		Now the statement for~$\lambda_1$ follows from the fact that there are~$m$ distinct orbits and $\binom{m}{2}$ pairs thereof.

		It remains to compute~$\lambda_0$. For each of the~$m$ orbits we have $\lambda_{0,j,j}=(q^n-1)/(q-1)-1-\lambda_{1,j,j}$; see also \cref{theo:SidonCount}.
		On the other hand, the intersection distribution in \cref{def:DistDistrMultOrbits} takes $(q^n-1)/(q-1)$ intersections between distinct orbits into account.
		Hence for each of the $\binom{m}{2}$ pairs of distinct orbits we have $\lambda_{0,j,j'}=(q^n-1)/(q-1)-\lambda_{1,j,j'}$.
		Now the result for $\lambda_0$ follows.
	\end{proof}
	
%%%%%%%%%%%%%%%%%%%%%%%%
	\section{Conclusion and Open Problems}\label{sec:futureWork}
	
	In this paper we investigated the intersection distribution, and thus the distance distribution, for cyclic orbit codes that have maximum possible length and distance at least $2k-4$.
	For distance $2k-2$ the intersection distribution can be fully described and in fact depends only on $q,n,k$,
	while for distance $2k-4$ the additional parameter~$f=f(U)$ plays a role. Many cases remain to be investigated.
	We conclude with some specific open problems and directions for future work.
	
	\begin{arabiclist}
		\item Throughout our work, the parameter $f=|\cF(U)|$ plays a prominent role. Can we provide more information about~$f$ for more general subspaces?
        For instance, can we find a lower bound on~$f$ that guarantees distance $2k-4$?
       The question of how many fractions a subspace has may also be related to questions raised in \cite{RRT18} about the size of the ``product'' space $U^2=\langle \sum_{i=1}^n u_iv_i \mid u_i,v_i \in U \rangle$.
		
		\item Can we determine for given parameters $(q,n,k)$ the range for~$r$ (or~$\lambda_2$)  in \cref{thm:Dist2k4L2Structure}?  \cref{tab:rValues,tab:rValuesExhaustive} in \cref{sec:GeneralOrbits} show that the upper bound for $\lambda_2$ given in \cref{cor:2k-4DistIneq} is in most cases very poor.
				
		\item Our main tool for proving \cref{thm:Dist2k4L2Structure} was a detailed study of intersections $U\cap\alpha U$
		 of maximal dimension~$\ell=2$.  However, to find the intersection distribution for $\ell \geq 3$, studying intersections of maximal dimension is insufficient. What can we say about intersections that are not of maximal dimension?

	    \item Can \cref{thm:Dist2k4L2Structure} be generalized to cyclic subspace codes with multiple orbits and distance $2k-4$?
	\end{arabiclist}
	
	%%%%%%%%%%%%%%%%%%%%%%%%%%%%%%%%%%%%%%%%%%%%%%%%%%%%%%
	% Bibliography
	%%%%%%%%%%%%%%%%%%%%%%%%%%%%%%%%%%%%%%%%%%%%%%%%%%%%%%


\begin{thebibliography}{10}
		
		\bibitem{BEGR16}
		E.~Ben-Sasson, T.~Etzion, A.~Gabizon, and N.~Raviv.
		\newblock Subspace {P}olynomials and {C}yclic {S}ubspace {C}odes.
		\newblock {\em IEEE Trans. Inform. Theory}, IT-62(3):1157--1165, March 2016.
		
		\bibitem{ChenLiu18}
		B.~Chen and H.~Liu.
		\newblock Constructions of cyclic constant dimension codes.
		\newblock {\em Des. Codes Cryptogr.}, 86(6):1267--1279, July 2018.
		
		\bibitem{CKMP19}
		A.~Cossidente, S.~Kurz, G.~Marino, and F.~Pavese.
		\newblock Combining subspace codes.
		\newblock Preprint 2019. arXiv: 1911.03387, 2019.
		
		\bibitem{EtVa11}
		T.~Etzion and A.~Vardy.
		\newblock Error-correcting codes in projective space.
		\newblock {\em IEEE Trans. Inform. Theory}, IT-57:1165--1173, 2011.
		
		\bibitem{GLMT15}
		H.~Gluesing-Luerssen, K.~Morrison, and C.~Troha.
		\newblock Cyclic orbit codes and stabilizer subfields.
		\newblock {\em Adv. Math. Commun.}, 9(2):177--197, 2015.
		
		\bibitem{GLT16}
		H.~Gluesing-Luerssen and C.~Troha.
		\newblock Construction of subspace codes through linkage.
		\newblock {\em Adv. Math. Commun.}, 10:525--540, 2016.
		
		\bibitem{GPSV18}
		M.~Greferath, M.~Pav{\v{c}}evi{\'c}, N.~Silberstein, and M.~V{\'a}zques-Castro.
		\newblock {\em Network Coding and Subspace Designs}.
		\newblock Springer, 2018.
		
		\bibitem{HKKW16}
		D.~Heinlein, M.~Kiermaier, S.~Kurz, and A.~Wassermann.
		\newblock Tables of subspaces codes.
		\newblock Preprint 2016. arXiv: 1601.02864.
		
		\bibitem{HeKu17}
		D.~Heinlein and S.~Kurz.
		\newblock Coset construction for subspace codes.
		\newblock {\em IEEE Trans. Inform. Theory}, IT-63:7651--7660, 2017.
		
		\bibitem{KoKsch08}
		R.~Koetter and F.~R. Kschischang.
		\newblock Coding for errors and erasures in random network coding.
		\newblock {\em IEEE Trans. Inform. Theory}, IT-54:3579--3591, 2008.
		
		\bibitem{OtOz17}
		K.~Otal and F.~\"{O}zbudak.
		\newblock Cyclic {S}ubspace {C}odes via {S}ubspace {P}olynomials.
		\newblock {\em Des. Codes Cryptogr.}, 85(2):191--204, Nov. 2017.
		
		\bibitem{RRT18}
		R.~M. Roth, N.~Raviv, and I.~Tamo.
		\newblock Construction of {S}idon {S}paces with {A}pplications to {C}oding.
		\newblock {\em IEEE Trans. Inform. Theory}, IT-64(6):4412--4422, June 2018.
		
		\bibitem{SiTr15}
		N.~Silberstein and A.-L. Trautmann.
		\newblock Subspace codes based on graph matchings, {F}errers diagrams, and
		pending blocks.
		\newblock {\em IEEE Trans. Inform. Theory}, IT-61:3937--3953, 2015.
		
		\bibitem{HKK14}
		{T. Honold and M. Kiermaier and S. Kurz}.
		\newblock Optimal binary subspace codes of length~$6$, constant dimension~$3$
		and minimum subspace distance~$4$.
		\newblock {\em Contemp. Math.}, 632:157--–176, 2015.
		
		\bibitem{TMBR13}
		A.~L. Trautmann, F.~Manganiello, M.~Braun, and J.~Rosenthal.
		\newblock Cyclic {O}rbit {C}odes.
		\newblock {\em IEEE Trans. Inform. Theory}, IT-59(11):7386--7404, Nov 2013.
		
		\bibitem{ZhTa19}
		W.~Zhao and X.~Tang.
		\newblock A characterization of cyclic subspace codes via subspace polynomials.
		\newblock {\em Finite Fields Appl.}, 57:1--12, 2019.
		
	\end{thebibliography}
\end{document}